\pdfoutput=1
\documentclass[11pt,a4paper]{article}
\usepackage[margin=1in]{geometry}
\usepackage[T1]{fontenc}
\usepackage{lmodern}
\usepackage{amsmath,amssymb,amsthm}
\usepackage{booktabs}
\usepackage{array}
\usepackage{tikz}
\usetikzlibrary{decorations.pathreplacing}
\usepackage[ruled,vlined,linesnumbered]{algorithm2e}
\usepackage[hidelinks]{hyperref}

\newtheorem{theorem}{Theorem}
\newtheorem{lemma}{Lemma}
\newtheorem{corollary}{Corollary}
\theoremstyle{definition}
\newtheorem{definition}{Definition}
\newtheorem{example}{Example}

\newcommand{\cL}{\mathcal{L}}
\newcommand{\cT}{\mathcal{T}}
\newcommand{\cD}{\mathcal{D}}
\newcommand{\Gint}[2]{G_{[#1,#2]}}
\newcommand{\Np}{N^{+}}
\newcommand{\first}{\mathrm{first}}
\newcommand{\poly}{\mathrm{poly}}

\begin{document}

\title{Tight bounds and output sensitive algorithms for maximal clique enumeration in link streams}
\author{George Manoussakis\\[2pt]
{\small LI-PARAD, Universit\'e de Versailles Saint-Quentin-en-Yvelines, Universit\'e Paris-Saclay, France}\\
{\small\texttt{george.manoussakis@uvsq.fr}}}
\date{}

\maketitle

\begin{abstract}
Given a link stream, we are interested in the question of finding all its maximal cliques. A link stream is a set of interactions between pairs of vertices, each one lasting during some time interval, and a clique of a link stream is a set of vertices together with a time interval during which all of them interact. A clique is maximal if neither its vertex set nor its interval can be enlarged. This question models the detection of groups in temporal data and has been studied for about ten years, and several algorithms list maximal cliques efficiently in practice. All known algorithms may however spend time exponential in the size of a clique for each clique they output, and the best known bounds on the number of maximal cliques carry a factor exponential in the degeneracy of the instantaneous graphs that is larger than in static graphs. Here, we approach the question under the light of two parameters of the instantaneous graphs of the stream, their maximum degree $\Delta$ and their degeneracy $k$. We first prove that a link stream with $m$ links has $O(m\Delta^2 3^{\Delta/3})$ maximal cliques, and that a link stream with $n$ vertices and $|\cT|$ distinct end times has $O(n|\cT|k^2 3^{k/3})$ maximal cliques. Both bounds are tight up to a factor polynomial in $\Delta$ and $k$ respectively, and the second one improves the factor $2^k$ of previous bounds to $3^{k/3}$. We also show that no bound of the form $m\cdot f(k)$ can hold. Then we present two output sensitive algorithms. The first one has setup time $O(m\log m)$ and polynomial time delay $\poly(\Delta)\log m$. The second one has setup time $O(m\log m+mk^2\log^3 n)$ and polynomial time delay $\poly(k)\log m$, and does not depend on the maximum degree. To the best of our knowledge, these are the first algorithms with polynomial time delay for this problem, in general link streams with durations and with both the vertex set and the interval maximal. Finally, we show that an online version of the second algorithm outputs every maximal clique at its end time in total time $O(m\log m+mk\log^3n)+\alpha\poly(k)\log n$, where $\alpha$ is the number of maximal cliques, and that all our results apply to the $\Delta$-cliques of temporal graphs. Our tools are a threshold characterization of maximal cliques, the decomposition by degeneracy orderings used for static graphs, a suffix argument that decides maximality without inspecting the neighborhood of a vertex, and dynamic orientations of graphs of bounded arboricity.
\end{abstract}


\section{Introduction}

Many data sets record interactions between entities together with the moment at which they happen and the time during which they last. For instance, phone calls, face-to-face contacts measured by sensors, messages, or co-presence in a place. Such data are naturally modeled by link streams, also called temporal graphs or temporal networks \cite{HS12,CFQS12,Mic16,LVM18}. Some of these formalisms use discrete time steps and instantaneous interactions, others continuous time and interactions with durations. Many problems which are polynomial on static graphs become NP-hard in the temporal setting, even on very restricted inputs, and the thesis of Molter \cite{Mol20} gives a parameterized overview. We use link streams with durations \cite{LVM18,VML18,BMT24}, which contain the $\Delta$-cliques of discrete temporal graphs as a special case (Section~\ref{sec:delta}). In a link stream, a link $(b,e,u,v)$ says that the vertices $u$ and $v$ interact during the time interval $[b,e]$. At every time $t$, the links present at $t$ form a static graph called the instantaneous graph at time $t$. Real link streams are often sparse in the sense that their instantaneous graphs have small maximum degree or small degeneracy, which is the smallest integer $k$ such that every subgraph contains a vertex of degree at most $k$. For instance, Bentert et al.\ \cite{BHM+19} measured a temporal version of the degeneracy on many real-world temporal graphs and observed that it is small, and Baudin et al.\ \cite{BMT24} observed that the maximum degree of the instantaneous graphs varies in a much smaller range than the number of vertices and links.

Groups of entities that all interact with each other during some period are a basic object of study in these data. Viard, Latapy and Magnien introduced the corresponding notion of clique \cite{VLM15,VLM16}. A clique of a link stream is a pair $(C,[t_0,t_1])$ where $C$ is a set of at least two vertices and every pair of vertices of $C$ interacts during the whole interval $[t_0,t_1]$. It is maximal if no vertex can be added to $C$ and if the interval cannot be extended. Maximal cliques have been used to reveal contact patterns among high-school students \cite{VLM15} and to define communities in link streams \cite{BMT23}. It thus seems interesting both from a practical and theoretical point of view to study the enumeration of maximal cliques in link streams, and to take into account the sparsity of their instantaneous graphs.

\paragraph*{Maximal cliques in static graphs} For static graphs, the question has been extensively studied and we can essentially distinguish between two families of algorithms. On one side, worst-case output size algorithms have been proposed. Their complexities match the maximum number of maximal cliques one can find in the considered graphs. For instance, Tomita et al.\ \cite{TTT06} propose an algorithm enumerating all maximal cliques of an $n$-order graph in time $O(3^{n/3})$, which is optimal since the Moon--Moser graphs have $3^{n/3}$ maximal cliques \cite{MM65}. Similarly, for $k$-degenerate graphs, Eppstein et al.\ \cite{ELS13} prove an $O((n-k)3^{k/3})$ bound on the number of maximal cliques, which is attained, and give an algorithm running in time $O(k(n-k)3^{k/3})$, which is fast on large real-world graphs \cite{ES11}. Both algorithms are based on the Bron--Kerbosch algorithm \cite{BK73}. On the other side, output sensitive algorithms bound the time spent for each maximal clique by a polynomial in the parameters of the graph, and many of them have polynomial time delay, that is, the wait between the output of two maximal cliques is polynomial. Tsukiyama et al.\ \cite{TIAS77} gave the first such algorithm, with delay $O(nm)$, and Johnson et al.\ \cite{JYP88} showed how to list the maximal cliques in lexicographic order with polynomial delay. Chiba and Nishizeki \cite{CN85} listed them in time $O(am)$ per clique, where $a$ is the arboricity. Makino and Uno \cite{MU04} obtained a delay of order the fourth power of the maximum degree and, with fast matrix multiplication, $O(n^{\omega})$; Comin and Rizzi \cite{CR18} improved the latter with rectangular matrix multiplication. Chang et al.\ \cite{CYQ13} gave an output sensitive algorithm for sparse graphs analysed with the $h$-index of the graph, and Conte et al.\ \cite{CGMV20} obtained a delay depending on the degeneracy, the maximum degree and the maximum clique size, in space linear in the graph. Conte and Tomita showed that the delay of the Bron--Kerbosch algorithm with pivoting can be exponential \cite{CT22}. The author proposed a decomposition of a $k$-degenerate graph into $n$ subgraphs of at most $k+1$ vertices which yields an enumeration time depending only on the degeneracy \cite{Man19}. A clique which is maximal in a subgraph but not in the graph is a proper suffix of a maximal clique found before, so that such cliques can be rejected with a generalized suffix tree; we use the same observation in Section~\ref{sec:degeneracy}. A delay depending only on the degeneracy can then be obtained by buffering the solutions \cite{ManDelay}. Several of these algorithms, and ours, first guarantee a bound on the time needed to produce the first $i$ solutions, for every $i$, and then turn it into a delay bound by buffering solutions in a queue. This technique goes back to Uno \cite{Uno15}. Capelli and Strozecki analysed it and gave a geometric amortization scheme which achieves a similar delay, up to a logarithmic factor, in polynomial space \cite{CS23,CS26}. Related enumeration problems include $k$-plexes \cite{BCK15,JM25} and $k$-defective cliques \cite{DLLW23,JM25}.

\paragraph*{Maximal cliques in link streams} The known results are summarized in Tables~\ref{tab:sota} and~\ref{tab:counts}, and we present them in chronological order. Viard, Latapy and Magnien first defined $\Delta$-cliques in link streams in which interactions are instantaneous \cite{VLM15,VLM16}. In a $\Delta$-clique $(X,[a,b])$, every pair of vertices of $X$ interacts at least once in every subinterval of $[a,b]$ of length $\Delta$. We keep the name $\Delta$-clique used in the literature, but we denote by $\theta$ the length of these subintervals, since $\Delta$ denotes the maximum degree in this paper. Viard et al.\ gave an algorithm which grows cliques in time and in vertices from each link, in time exponential in $n$ \cite{VLM16}, and they then extended it to links with durations \cite{VML18}.

Himmel et al.\ \cite{HMNS16,HMNS17} adapted the Bron--Kerbosch algorithm with pivoting to $\Delta$-cliques. They introduced the $\Delta$-slice degeneracy, which is the largest degeneracy of the graphs formed by the interactions of a window of length $\theta$, and showed that the problem is fixed-parameter tractable with respect to this parameter. The running time of their algorithm is bounded in terms of the number of time-maximal cliques, which are the cliques whose interval cannot be extended but whose vertex set might be. The effect of pivoting on the worst-case running time was left open. Bentert et al.\ \cite{BHM+18,BHM+19} extended this algorithm to $\Delta$-$k$-plexes. For cliques, they bounded the number of time-maximal cliques by $n2^{k+1}\min(m,|T|)$, where $k$ is the $\Delta$-slice degeneracy and $|T|$ the number of time steps. We use the letter $k$ for the $\Delta$-slice degeneracy since, by Section~\ref{sec:delta}, it is the largest degeneracy of the instantaneous graphs of a link stream associated with the temporal graph. This improves the bound $3^{k/3}2^{k+1}n|T|$ of \cite{HMNS17}. These results are presented together in the thesis of Molter \cite{Mol20}. Molter et al.\ studied isolated temporal cliques, which have few links to the rest of the stream, and gave fixed-parameter algorithms with respect to the isolation parameter \cite{MNR19,MNR21}. They note that their algorithm is not output sensitive.

Banerjee and Pal required every pair to interact at least $\gamma$ times in every window of length $\theta$. They gave an algorithm listing the maximal $(\Delta,\gamma)$-cliques in time $O(2^nn^2m^3+2^nn^3m^2)$ \cite{BP19}, then a two-phase algorithm which first stretches cliques of two vertices in time and then adds vertices \cite{BP21,BP24}, and an incremental algorithm \cite{BP22}. Pal adapted the algorithm of Himmel et al.\ to $(\Delta,\gamma)$-cliques, with running time $O(2^k3^{k/3}n|T|m+n^2|T|^2)$, where $n$ and $m$ are the numbers of vertices and edges and $k$ is the $\Delta$-slice degeneracy, which ignores $\gamma$ \cite{Pal24}. Boekhout and Takes extended the two-phase approach to weighted temporal networks, in which $\gamma$ bounds the total weight of the interactions of a pair in a window \cite{BT25}. They perform the first phase in time $O(m)$ instead of $O(\gamma m)$, and their algorithm is much faster in practice. They also require the interval of a clique to start and end with an interaction between two of its vertices, which changes the set of maximal cliques.

Baudin et al.\ \cite{BMT24,Bau23} gave an algorithm which, for every time $t$ at which a link starts, lists the cliques of the instantaneous graph at $t$ which contain this link, and keeps those which are vertex-maximal. It handles link streams with more than $10^8$ links. Its running time is $O(m3^{\Delta/3}2^q\Delta^2)$, where $\Delta$ and $q$ are the largest maximum degree and clique size of an instantaneous graph. Its output sensitive running time is $O(\frac1r\Delta^2q\alpha)$, where $\alpha$ is the number of maximal cliques and $r$ the proportion of leaves of the search tree which output a maximal clique. The authors show that $1\le\frac1r\le 2^q$ and that $\frac1r=2^{q-1}-1$ can occur without pivoting, while with their pivot $\frac1r$ is below 2 in their experiments. The thesis of Baudin \cite{Bau23} states the understanding of this factor as a perspective.

Davot et al.\ \cite{DEMM25} consider the same maximal $\Delta$-cliques as we do, in temporal graphs which satisfy two assumptions. The first one is a temporal version of the closure of static graphs: two vertices with at least $\rho$ common neighbors during a short interval of time are adjacent around that time. The second one is that the neighborhood of every vertex changes by at most $\eta$ vertices between two consecutive time steps. They prove that the number of maximal $\Delta$-cliques is then at most $3\cdot2^{\rho-1+2\eta(\theta+1-\theta_1)}n^2|T|$, where $\theta_1\le\theta$ is the length of the short intervals, and they list them with the algorithm of \cite{HMNS17}. Their parameters and the degeneracy are incomparable. For instance, a complete graph present at every time step has closure number $1$ and degeneracy $n-1$, while $K_{2,N}$ has degeneracy $2$ and closure number $N+1$. Observe that their bound also bounds the number of maximal cliques output by our algorithms on these temporal graphs. Finally, a bound in terms of the vertex deletion distance to order preservation is implicit in the work of Hermelin et al.\ \cite{HIMN23}, as observed in \cite{DEMM25}.

To summarize, the analyses of \cite{HMNS17,BHM+19,BMT24} bound the work by the number of time-maximal cliques, and this number can exceed the number of maximal cliques by a factor exponential in the size of a clique. As a consequence, none of these algorithms is known to spend time polynomial in the input for each maximal clique it outputs. Moreover, the best known bound on the number of maximal cliques \cite{BHM+19} carries a factor $2^k$, while in static graphs the corresponding factor is $3^{k/3}$ \cite{ELS13}.

\paragraph*{Cliques maximal over the whole lifespan} Brunelli et al.\ \cite{BCGM25} recently gave an output sensitive algorithm for another notion of maximal clique. For them, a temporal clique is a set of vertices which is a clique of some instantaneous graph, and it is maximal if no strict superset of it is a clique of some instantaneous graph. Their solutions are thus vertex sets, without time interval. This family of sets is hereditary, which allows them to use the framework of Cohen et al.\ \cite{CKS08}. They consider single-interval temporal graphs, in which every pair interacts during at most one interval of time. In such a graph with $n$ vertices and $m$ edges, they list the $\alpha'$ maximal temporal cliques in time $O(\alpha'^2nq(q^3+\Delta))$ and space $O(mn)$, where $q$ is the maximum size of a temporal clique and $\Delta$ the maximum degree of the graph of all the pairs which interact. They also list the first $i$ of them in time $O(i^2q(q^3+\Delta))$, using space $O(mn+qi)$.

The two problems are related as follows. If $S$ is a solution of \cite{BCGM25}, then $S$ together with the interval during which it is a clique is a maximal clique in our sense. Thus $\alpha'\le\alpha$. The converse does not hold. For instance, in Figure~\ref{fig:stream}, $(\{a,b\},[1,8])$ is a maximal clique although $\{a,b\}$ is contained in the clique $\{a,b,c\}$. Moreover, $\alpha$ can be exponentially larger than $\alpha'$. Let $p\ge2$ and consider the complete $p$-partite graph whose parts have three vertices each. Give every edge of this graph the link $(0,2,u,v)$, and every pair inside a part the link $(1,2,u,v)$. Then the only solution of \cite{BCGM25} is the whole vertex set, while the link stream has $3^p+1$ maximal cliques. This example is the one used in \cite{BCGM25} to compare their solutions with the maximal cliques of the instantaneous graphs. Therefore neither result implies the other. Observe also that, in our sense, the maximality of a clique only depends on the links present during its interval. On the contrary, a set of vertices which interacts during a long period is not a solution of \cite{BCGM25} as soon as a larger set interacts once at some other time. Finally, the reduction from $\Delta$-cliques given in \cite{BCGM25} produces temporal graphs in which a pair may interact during several intervals, to which their algorithm does not directly apply.

Other dense structures have been studied in temporal graphs. Qin et al.\ studied periodic cliques, whose edges reappear at regular intervals \cite{QLW+19}, and Zhu et al.\ listed sets of $k$ time intervals that pairwise overlap within a query window, which they call temporal $k$-cliques \cite{ZFY+19}. Temporal cores \cite{WCL+15}, span-cores \cite{GBB+18}, persistent communities \cite{LSQ+18} and bursting cores \cite{QLY+22} have also been considered, and Zhao et al.\ listed the maximal $\Delta$-bicliques of temporal bipartite graphs \cite{ZCY+25}. The name temporal clique is also used for a temporal graph whose underlying graph is complete, in the study of temporal spanners \cite{CPS19}, which is a different object.

\begin{table}[t]
\centering\small
\renewcommand{\arraystretch}{1.2}
\begin{tabular}{@{}>{\raggedright}p{3.4cm}>{\raggedright}p{2.0cm}>{\raggedright}p{3.3cm}>{\raggedright\arraybackslash}p{5.3cm}@{}}
\toprule
Algorithm & Model & Setup & Enumeration\\
\midrule
Viard et al.\ \cite{VLM16} & $\Delta$-cliques & & $O(2^nn^2m^2(n+\log m))$\\
Himmel et al.\ \cite{HMNS17} & $\Delta$-cliques & & $O(3^{k/3}2^kn|T|m)$ and $O(xm+m|T|)$\\
Viard et al.\ \cite{VML18} & durations & & $O(2^nn^3m^2+2^nn^2\log m)$\\
Bentert et al.\ \cite{BHM+19} & $\Delta$-$k$-plexes & & $O(2^{k}\min(m,|T|)^2n^4)$ for cliques; FPT with respect to $k$, output sensitive in $x$\\
Molter et al.\ \cite{MNR19,MNR21} & isolated $\Delta$-cliques & & FPT with respect to the isolation parameter; not output sensitive\\
Banerjee and Pal \cite{BP19} & $(\Delta,\gamma)$-cliques & & $O(2^nn^2m^3+2^nn^3m^2)$\\
Banerjee and Pal \cite{BP21,BP22,BP24} & $(\Delta,\gamma)$-cliques & & bounds in terms of the input; experiments\\
Pal \cite{Pal24} & $(\Delta,\gamma)$-cliques & & $O(2^k3^{k/3}n|T|m+n^2|T|^2)$\\
Boekhout and Takes \cite{BT25} & weighted $(\Delta,\gamma)$-cliques, interactions at both ends & & first phase in $O(m)$; experiments\\
Baudin et al.\ \cite{BMT24} & durations & & $O(m3^{\Delta/3}2^q\Delta^2)$ and $O(\frac1r\Delta^2q\alpha)$, $1\le\frac1r\le 2^q$\\
Brunelli et al.\ \cite{BCGM25} & one interval per pair; vertex sets & & $O(\alpha'^2nq(q^3+\Delta))$, where $\alpha'\le\alpha$ is their number of solutions; first $i$ solutions in $O(i^2q(q^3+\Delta))$\\
\midrule
This paper$^{+}$, Thm.~\ref{thm:degree-enum} & durations & $O(m\log m)$ & $\alpha\,O(q^2\Delta^2(\Delta\log m+\Delta^2+\lambda(\Delta)))$\\
This paper$^{+}$, Thm.~\ref{thm:delay} & durations & $O(m\log m+mk^2\log^3n)$ & $\alpha\,O(k^5(k^2\log m+\lambda(k)))$\\
This paper (online), Thm.~\ref{thm:online} & durations & & $O(m\log m+mk\log^3n)+\alpha\,O(k^5(k^2\log n+\lambda(k)))$\\
\bottomrule
\end{tabular}

\smallskip
\raggedright\footnotesize
$\Delta$ = maximum degree, $k$ = degeneracy, $q$ = maximum clique size of an instantaneous graph (for $\Delta$-cliques, of a window graph), $\alpha$ = number of maximal cliques, $x$ = number of time-maximal cliques, $|T|$ = number of time steps, $\lambda(h)=O(h^3)$ = delay of maximal clique enumeration in an $h$-order graph. $+$: polynomial time delay algorithms, with time delay $O(\Delta^4(\Delta\log m+\Delta^2+\lambda(\Delta)))$ and $O(k^5(k^2\log m+\lambda(k)))$ respectively. The online algorithm has no setup phase, and its total running time is given.
\caption{Bounds for maximal clique enumeration in link streams and temporal graphs. The results are presented chronologically.}\label{tab:sota}
\end{table}

\begin{table}[t]
\centering\small
\renewcommand{\arraystretch}{1.2}
\begin{tabular}{@{}>{\raggedright}p{3.4cm}>{\raggedright}p{2.6cm}>{\raggedright}p{4.0cm}>{\raggedright\arraybackslash}p{4.0cm}@{}}
\toprule
Reference & Objects counted & Upper bound & Lower bound\\
\midrule
Himmel et al.\ \cite{HMNS17} & time-maximal $\Delta$-cliques & $3^{k/3}2^{k+1}n|T|$ & \\
Bentert et al.\ \cite{BHM+19} & time-maximal $\Delta$-cliques & $2^{k+1}n\min(m,|T|)$ & \\
Davot et al.\ \cite{DEMM25} & maximal $\Delta$-cliques, closure $\rho$, instability $\eta$ & $2^{O(\rho+\eta(\theta+1))}n^2|T|$ & \\
\midrule
This paper, Thm.~\ref{thm:degree-count} & maximal cliques & $O(m\Delta^23^{\Delta/3})$ & $\Omega(m3^{\Delta/3})$\\
This paper, Ex.~\ref{ex:quadratic} & maximal cliques & & $\Omega(m^2)$ with $k=2$\\
This paper, Thm.~\ref{thm:degeneracy-count} & maximal cliques & $O(n|\cT|k^23^{k/3})$ & $\Omega(n|\cT|3^{k/3})$\\
\bottomrule
\end{tabular}

\smallskip
\raggedright\footnotesize
$|\cT|$ = number of distinct end times of links, with $|\cT|\le\min(m,|T|)$ for $\Delta$-cliques. $|T|$ = number of time steps. For \cite{DEMM25}, $\theta$ is the length of the windows of the $\Delta$-cliques, and $\rho$ and $\eta$ are the closure and instability parameters. Every maximal clique is time-maximal, so the first two rows also bound the number of maximal $\Delta$-cliques.
\caption{Bounds on the number of maximal cliques of link streams and temporal graphs.}\label{tab:counts}
\end{table}

\paragraph*{Related questions} A different problem is to maintain the set of maximal cliques of a graph under edge insertions and deletions. Stix \cite{Sti04} and Ottosen and Vomlel \cite{OV10} gave update algorithms, and Das et al.\ gave the first change-sensitive algorithms, whose time is proportional to the number of maximal cliques that appear or disappear, together with nearly tight bounds on this number \cite{DST19}. Baudin et al.\ note that this approach can be adapted to link streams \cite{BMT24}. Observe however that the maximal cliques of a link stream are not the maximal cliques of its instantaneous graphs, since a vertex adjacent to all the vertices of a clique during part of its interval does not extend it. Inside one end time, our algorithms solve a problem of this incremental kind on a graph with $O(k)$ or $O(\Delta)$ vertices. We will also need to maintain orientations of the instantaneous graphs. An orientation has out-degree at most $c$ if every vertex has at most $c$ out-neighbors, and an acyclic orientation of out-degree at most $k$ exists if and only if the degeneracy is at most $k$. Brodal and Fagerberg maintained an orientation of out-degree $O(a)$ under edge insertions and deletions, where $a$ bounds the arboricity \cite{BF99}, Kopelowitz et al.\ obtained worst-case update times \cite{KKPS14}, and Christiansen and Rotenberg obtained out-degree $a+2$ \cite{CR22}. These orientations need not be acyclic. The level data structures of Bhattacharya et al.\ \cite{BHNT15} and Henzinger et al.\ \cite{HNW20} give acyclic orientations of out-degree $O(a)$, where $a$ is the current arboricity, with polylogarithmic amortized update time; Liu et al.\ stated this orientation explicitly \cite{LSY+22}.

\paragraph*{Our contributions} Let $\cL$ be a link stream with $n$ vertices and $m$ links, let $\cT$ be the set of distinct end times of its links, and let $\Delta$, $k$ and $q$ be the largest maximum degree, the largest degeneracy and the largest clique size of an instantaneous graph. We first prove that a maximal clique ending at time $e$ is determined by its start time $\tau$ and by a maximal clique of a static graph $\Gint{\tau}{e}$ whose edges are the pairs that interact throughout $[\tau,e]$. As $\tau$ increases, these graphs only gain edges. Therefore, all maximal cliques ending at $e$ can be found by sweeping $\tau$ over the start times of the links present at $e$, and by listing only maximal cliques of static graphs. This is what removes the factor exponential in the clique size of the previous algorithms.

Using this observation, we present a first output sensitive algorithm, which anchors every maximal clique at one of its links that ends at the end time of the clique. It has setup time $O(m\log m)$, amortized time delay $O(q^2\Delta^2(\Delta\log m+\Delta^2+\lambda(\Delta)))$, where $\lambda(h)$ is the delay of maximal clique enumeration in an $h$-order graph, and polynomial time delay $O(\Delta^4(\Delta\log m+\Delta^2+\lambda(\Delta)))$. For example, using the algorithm of Tsukiyama et al.\ \cite{TIAS77}, the delay is $O(\Delta^7+\Delta^5\log m)$. As a by-product, we prove that a link stream has $O(m\Delta^23^{\Delta/3})$ maximal cliques, and that this bound is tight up to a factor $O(\Delta^2)$. An online version of the algorithm outputs every maximal clique at its end time.

In a second part, we consider the degeneracy. We first observe that the number of maximal cliques cannot be bounded by $m$ times a function of $k$, since it can be quadratic in $m$ when $k=2$. We then prove that a link stream has $O(n|\cT|k^23^{k/3})$ maximal cliques and that this is tight up to a factor polynomial in $k$. This improves the factor $2^k$ of the bound of \cite{BHM+19}, and the factor $3^{k/3}2^k$ of the bound of \cite{HMNS17}, to $3^{k/3}$, which cannot be improved. We then present a second output sensitive algorithm, which anchors every maximal clique at its end time and at its first vertex in an ordering of the instantaneous graph in which every vertex has few later neighbors, as in the decomposition of \cite{Man19}. A candidate found in the neighborhood of an anchor may fail to be maximal because of a vertex that comes earlier in the ordering. We show that such a candidate is a proper suffix of a maximal clique that has already been output, so that it can be rejected by a dictionary lookup, without inspecting the neighborhood of the anchor. Finally, we maintain the orderings over time with the level data structures of \cite{BHNT15,HNW20}. The resulting algorithm has setup time $O(m\log m+mk^2\log^3n)$ and polynomial time delay $O(k^5(k^2\log m+\lambda(k)))$, which does not depend on the maximum degree. It can also be run online, in which case it outputs every maximal clique at its end time in total time $O(m\log m+mk\log^3n)+\alpha O(k^5(k^2\log n+\lambda(k)))$, where $\alpha$ is the number of maximal cliques. When the output dominates the setup phase, its worst-case running time matches the largest possible number of maximal cliques up to a factor $\poly(k)\log m$. All our results apply to the $\Delta$-cliques of temporal graphs, with $k$ equal to their $\Delta$-slice degeneracy. To the best of our knowledge, these are the first algorithms with polynomial time delay for listing the maximal cliques of link streams, in which both the vertex set and the interval are maximal. This problem differs from the enumeration of the vertex sets which are maximal over all times \cite{BCGM25}, and from enumeration under closure and stability assumptions \cite{DEMM25}.

The organization of the document is as follows. In Section~\ref{sec:defs}, we introduce some notations and definitions. In Section~\ref{sec:basic}, we prove basic results. These results are used in Sections~\ref{sec:degree} and~\ref{sec:degeneracy} for the proofs of the main contributions of the paper. The first algorithm and the bounds in terms of the maximum degree are presented in Theorems~\ref{thm:edge-correct}, \ref{thm:degree-count} and~\ref{thm:degree-enum}. The bounds in terms of the degeneracy are proved in Theorem~\ref{thm:degeneracy-count}, and the second algorithm in Theorems~\ref{thm:vertex-correct}, \ref{thm:vertex-given}, \ref{thm:online} and~\ref{thm:delay}. In Section~\ref{sec:delta}, we apply our results to $\Delta$-cliques. We conclude in Section~\ref{sec:conclusion}.

\section{Definitions}\label{sec:defs}

We consider graphs of the form $G=(V,E)$ which are simple and undirected. If $X\subseteq V$, the subgraph of $G$ induced by $X$ is denoted by $G[X]$. The set $N_G(x)$, or $N(x)$ when not ambiguous, is the open neighborhood of the vertex $x$ and consists of the vertices adjacent to $x$ in $G$. A clique is a set of pairwise adjacent vertices, and it is maximal if it is not strictly contained in another clique. A graph on $h$ vertices has at most $3^{h/3}$ maximal cliques \cite{MM65}. By $[n]$ we denote the set of integers $\{1,2,\dots,n\}$.

A graph is $k$-degenerate if there is an ordering $v_1,\dots,v_n$ of its vertices in which every vertex $v_i$ has at most $k$ neighbors in $\{v_{i+1},\dots,v_n\}$, and its degeneracy is the smallest such $k$. Such an ordering is called a degeneracy ordering and can be computed in linear time \cite{MB83}. Given an ordering $\sigma$ of the vertices, $\sigma(x)$ is the rank of the vertex $x$ (the index of its position) in $\sigma$.

In our algorithms, we will need to consider the vertices of some cliques as strings. A string $s$ is a suffix of a string $t$ if there exists a string $u$ such that $t=us$. If $u$ is not empty, $s$ is called a proper suffix of $t$.

\begin{definition}
A link stream is a triple $\cL=(T,V,E)$ where $T\subseteq\mathbb{R}$ is an interval, $V$ is a finite set of $n$ vertices and $E$ is a finite set of $m$ links $(b,e,u,v)$ with $b,e\in T$, $b\le e$ and $u\neq v$. The pair $\{u,v\}$, written $uv$, is unordered.
\end{definition}

As in \cite{BMT24}, we assume that the link stream is simple, that is, two links of the same pair have disjoint intervals. Links of the same pair whose intervals intersect can always be merged into one link, since this does not change the times at which the pair interacts. In a simple link stream, a pair interacts at every time of an interval if and only if one of its links contains this interval. We also assume that there are no isolated vertices, so that $n\le 2m$.

\begin{definition}
For a time $t$, the instantaneous graph $G_t$ has vertex set $V$ and an edge $uv$ whenever some link $(b,e,u,v)$ satisfies $b\le t\le e$. For an interval $[t_0,t_1]$, the graph $\Gint{t_0}{t_1}$ has vertex set $V$ and an edge $uv$ whenever some link of $uv$ contains $[t_0,t_1]$.
\end{definition}

If $uv$ is an edge of $G_t$, the link of $uv$ containing $t$ is unique since the stream is simple, and we denote its start and its end by $b_t(uv)$ and $e_t(uv)$.

\begin{definition}
A clique of $\cL$ is a pair $(C,[t_0,t_1])$ with $|C|\ge 2$ such that $C$ is a clique of $\Gint{t_0}{t_1}$. A clique $(C,I)$ is maximal if there is no other clique $(C',I')$ with $C\subseteq C'$ and $I\subseteq I'$.
\end{definition}

This is equivalent to requiring that the interval cannot be extended and that no vertex can be added, which are called time-maximality and vertex-maximality in \cite{BMT24}. We denote by $\alpha$ the number of maximal cliques of $\cL$, by $\cT$ the set of distinct end times of links, and, for $e\in\cT$, by $P_e$ the set of pairs having a link that ends at $e$, called the ending pairs at $e$. We let $\Delta$, $k$ and $q$ be the largest maximum degree, the largest degeneracy and the largest clique size of a graph $G_t$, $t\in T$. These maxima are taken over finitely many graphs, since $G_t$ only changes at the start and end times of links, and $q\le k+1\le\Delta+1$.

We use the RAM model, and an algorithm must output every solution explicitly. As in \cite{Man19}, the time complexity of an enumeration algorithm is divided into a setup phase followed by an enumeration phase. Let $S_1,\dots,S_\alpha$ be the solutions in the order in which they are output, let $T(i)$ be the number of steps performed after the setup phase and before $S_i$ is output, for $i\in[\alpha]$, and let $T(\alpha+1)$ be the number of steps performed after the setup phase until the algorithm stops. Set $T(0)=0$. The algorithm has time delay $p$ if $T(i+1)-T(i)\le p$ for all $0\le i\le\alpha$, and amortized time delay $p$ if $T(i)\le ip$ for all $i\in[\alpha+1]$, following \cite{CS23}. The next lemma is folklore; it goes back to Uno \cite{Uno15} and is analysed in \cite{CS23}.

\begin{lemma}\label{lem:buffer}
An algorithm with amortized time delay $p$ whose solutions have size at most $p$, where $p$ is known at the end of the setup phase, can be turned into an algorithm with time delay $O(p)$ which outputs the same solutions after the same setup phase. The new algorithm may require space proportional to the total size of the solutions.
\end{lemma}

We will obtain amortized time delays with the following lemma. An algorithm examines a sequence of candidates; some of them are output and the others are rejected. A charging map sends every candidate to a solution.

\begin{lemma}\label{lem:accounting}
Assume that, after its setup phase, an algorithm examines candidates and that there are integers $\kappa$ and $\mu$ such that (i) at most $\kappa$ steps are performed before the first candidate is examined, between two consecutive candidates, and after the last one; (ii) every candidate is sent by a charging map to a solution that is output when this candidate is examined or before; (iii) at most $\mu$ candidates are sent to each solution. Then the algorithm has amortized time delay $(\mu+1)\kappa$.
\end{lemma}

\begin{proof}
Let $i\in[\alpha]$. By (ii), every candidate examined until $S_i$ is output is sent to one of $S_1,\dots,S_i$. By (iii), this implies that there are at most $i\mu$ such candidates, and thus $T(i)\le i\mu\kappa$ by (i). In the same way, the total number of candidates is at most $\alpha\mu$, and thus $T(\alpha+1)\le(\alpha\mu+1)\kappa\le(\alpha+1)(\mu+1)\kappa$.
\end{proof}

\section{Basic results}\label{sec:basic}

We want to show in this section that, roughly, a maximal clique of a link stream is determined by its end time, by its start time, and by a maximal clique of a static graph. Let $C$ be a clique of $G_t$ with $|C|\ge 2$. We define the start and the end of $C$ at time $t$ as
\[
s_t(C)=\max_{x,y\in C,\ x\neq y} b_t(xy)\qquad\text{and}\qquad f_t(C)=\min_{x,y\in C,\ x\neq y} e_t(xy).
\]

\begin{lemma}\label{lem:window}
Let $C$ be a clique of $G_t$ with $|C|\ge 2$ and let $I=[t_0,t_1]$ with $t_0\le t\le t_1$. Then $C$ is a clique of $G_I$ if and only if $I\subseteq[s_t(C),f_t(C)]$.
\end{lemma}

\begin{proof}
Let $x,y\in C$ be distinct. A link of $xy$ that contains $I$ contains $t$. Since the stream is simple, this implies that it is the link of $xy$ containing $t$, whose interval is $[b_t(xy),e_t(xy)]$. Thus $xy$ is an edge of $G_I$ if and only if $I\subseteq[b_t(xy),e_t(xy)]$, and the claim follows by considering all the pairs of $C$.
\end{proof}

Given a time $e$ and a threshold $\tau\le e$, Lemma~\ref{lem:window} applied to pairs implies that the edges of $\Gint{\tau}{e}$ are exactly the edges $xy$ of $G_e$ with $b_e(xy)\le\tau$. Therefore, for $\tau'\le\tau\le e$, graph $\Gint{\tau'}{e}$ is obtained from $\Gint{\tau}{e}$ by deleting edges. In other words, the graphs $\Gint{\tau}{e}$ gain edges as $\tau$ increases.

\begin{lemma}\label{lem:charact}
Let $e$ be a time, $\tau\le e$ and $C\subseteq V$ with $|C|\ge 2$. Then $(C,[\tau,e])$ is a maximal clique of $\cL$ if and only if (a) $C$ is a maximal clique of $\Gint{\tau}{e}$, (b) $s_e(C)=\tau$, and (c) $f_e(C)=e$, that is, some pair of $C$ is an ending pair at $e$. Conditions (b) and (c) are well defined under (a), since (a) implies that $C$ is a clique of $G_e$.
\end{lemma}

\begin{proof}
Assume first that $(C,[\tau,e])$ is a maximal clique. Then $C$ is a clique of $\Gint{\tau}{e}$, hence of $G_e$, and by Lemma~\ref{lem:window}, $(C,[s_e(C),f_e(C)])$ is a clique which contains $(C,[\tau,e])$. By maximality the two intervals are equal, which gives (b) and (c). If $C$ was not a maximal clique of $\Gint{\tau}{e}$, then a larger clique $C'$ of $\Gint{\tau}{e}$ would give a clique $(C',[\tau,e])$ containing $(C,[\tau,e])$, a contradiction. This gives (a).

Conversely, assume (a), (b) and (c). Then $(C,[\tau,e])$ is a clique. Let $(C',I')$ be a clique with $C\subseteq C'$ and $[\tau,e]\subseteq I'$. Observe that $C$ is a clique of $G_{I'}$ and that $e\in I'$. Thus Lemma~\ref{lem:window} implies that $I'\subseteq[s_e(C),f_e(C)]=[\tau,e]$, and therefore $I'=[\tau,e]$. This implies that $C'$ is a clique of $\Gint{\tau}{e}$ containing $C$, and thus $C'=C$ by (a).
\end{proof}

Therefore, a maximal clique ending at $e$ is a clique $C$ of $G_e$ containing an ending pair, together with the threshold $s_e(C)$, such that $C$ is maximal in $\Gint{s_e(C)}{e}$. Figure~\ref{fig:stream} gives an example.

\begin{figure}[t]
\centering
\begin{tikzpicture}[x=0.95cm,y=0.55cm,every node/.style={font=\small}]
  \draw[->] (0,0) -- (10.6,0) node[right] {time};
  \foreach \t in {0,...,10} { \draw (\t,0.1) -- (\t,-0.1) node[below] {\t}; }
  \fill[blue!10] (2,0.4) rectangle (5,5.1);
  \node[left] at (0,1) {$cd$}; \draw[very thick] (3,1) -- (9,1);
  \node[left] at (0,2) {$bc$}; \draw[very thick] (0,2) -- (5,2);
  \node[left] at (0,3) {$ac$}; \draw[very thick] (2,3) -- (6,3);
  \node[left] at (0,4) {$ab$}; \draw[very thick] (1,4) -- (8,4);
  \draw[dashed] (2,0) -- (2,5.3) node[above] {$\tau=s_5(C)$};
  \draw[dashed] (5,0) -- (5,5.3) node[above] {$e=f_5(C)$};
\end{tikzpicture}
\caption{A link stream with one row per pair, in which each bar is a link. Its maximal cliques are $(\{a,b,c\},[2,5])$, $(\{a,b\},[1,8])$, $(\{a,c\},[2,6])$, $(\{b,c\},[0,5])$ and $(\{c,d\},[3,9])$. For $C=\{a,b,c\}$, the link of $bc$ ends at $e=5$ and the link of $ac$ starts at $\tau=2$, and $C$ is a maximal clique of $\Gint{2}{5}$. The set $\{b,c\}$ is a maximal clique of $\Gint{0}{5}$ with $s_5(\{b,c\})=0$.}\label{fig:stream}
\end{figure}
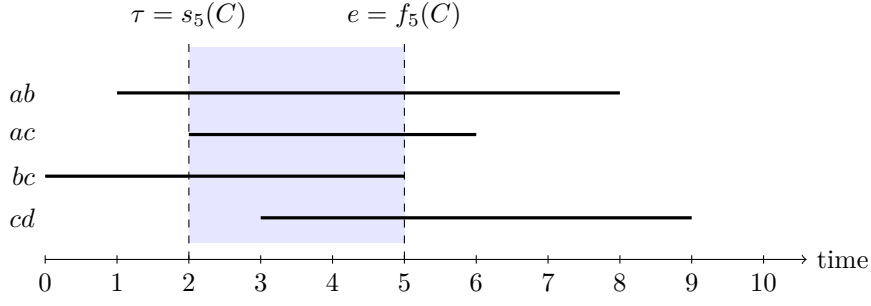

\begin{lemma}\label{lem:shadow}
Let $e$ be a time, $W\subseteq V$ and $\tau\le e$, and let $C\subseteq W$ be a maximal clique of $\Gint{\tau}{e}[W]$ with $|C|\ge 2$. Then $s_e(C)\le\tau$ and $C$ is a maximal clique of $\Gint{s_e(C)}{e}[W]$.
\end{lemma}

\begin{proof}
Every pair $xy$ of $C$ is an edge of $\Gint{\tau}{e}$, so $b_e(xy)\le\tau$, which implies $s_e(C)\le\tau$. Every pair $xy$ of $C$ also satisfies $b_e(xy)\le s_e(C)$, so $C$ is a clique of $\Gint{s_e(C)}{e}[W]$. Since this graph is obtained from $\Gint{\tau}{e}[W]$ by deleting edges, a clique of it which contains $C$ is a clique of $\Gint{\tau}{e}[W]$, and is thus equal to $C$.
\end{proof}

\section{Algorithms parameterized by the maximum degree}\label{sec:degree}

In this section we anchor every maximal clique at one of its links. We fix an arbitrary total order on the pairs of vertices. The canonical pair of a maximal clique $(C,[\tau,e])$ is the smallest pair of $C$ which is an ending pair at $e$; it exists by Lemma~\ref{lem:charact}(c). Its canonical link is the link of this pair which ends at $e$.

\begin{definition}\label{def:edge}
Let $\ell=(b,e,u,v)$ be a link. The set $H_\ell=N_{G_e}(u)\cap N_{G_e}(v)$ is the set of common neighbors of $u$ and $v$ at time $e$. The set $\Theta_\ell$ is the set of values $b_e(xy)$ over the edges $xy$ of $G_e$ with $x,y\in H_\ell\cup\{u,v\}$ and $b_e(xy)\ge b$. For $\tau\in\Theta_\ell$, graph $H^\tau_\ell$ is the subgraph of $\Gint{\tau}{e}$ induced by the common neighbors of $u$ and $v$ in $\Gint{\tau}{e}$.
\end{definition}

Observe that $b\in\Theta_\ell$, since $b_e(uv)=b$. The outline of the first algorithm is the following. We consider every link $\ell=(b,e,u,v)$, in increasing order of $e$ and then of the pair $uv$. For every threshold $\tau\in\Theta_\ell$, in increasing order, we find all the maximal cliques $K$ of $H^\tau_\ell$. For each of them, the set $C=K\cup\{u,v\}$ is a candidate. It is output with the interval $[\tau,e]$ if $s_e(C)=\tau$ and if $uv$ is its canonical pair, and rejected otherwise. The procedure is described in Algorithm~\ref{alg:edge}.

\begin{algorithm}[t]
\caption{Maximal cliques anchored at ending links}\label{alg:edge}
\KwData{A simple link stream $\cL$.}
\KwResult{All the maximal cliques of $\cL$.}
Sort the links by end time and then by pair\;
\For{every link $\ell=(b,e,u,v)$ in this order}{
  \For{every $\tau\in\Theta_\ell$ in increasing order}{
    \For{every maximal clique $K$ of $H_\ell^\tau$ (the empty set if $H^\tau_\ell$ has no vertex)}{
      $C\gets K\cup\{u,v\}$\;
      \eIf{$s_e(C)=\tau$ and $uv$ is the canonical pair of $C$}{Output $(C,[\tau,e])$.}{Reject it.}
    }
  }
}
\end{algorithm}

\begin{lemma}\label{lem:edge-cand}
For every candidate $C$ examined for the link $\ell=(b,e,u,v)$ and the threshold $\tau$, $C$ is a maximal clique of $\Gint{\tau}{e}$ and $(C,[s_e(C),e])$ is a maximal clique of $\cL$.
\end{lemma}

\begin{proof}
Since $b\le\tau$, the pair $uv$ is an edge of $\Gint{\tau}{e}$, and thus $C$ is a clique of $\Gint{\tau}{e}$. Assume by contradiction that a vertex $z\notin C$ is adjacent in $\Gint{\tau}{e}$ to all the vertices of $C$. Then $z$ is a common neighbor of $u$ and $v$ in $\Gint{\tau}{e}$, hence a vertex of $H^\tau_\ell$ adjacent to all the vertices of $K$. This contradicts the maximality of $K$, and $C$ is a maximal clique of $\Gint{\tau}{e}$. By Lemma~\ref{lem:shadow} with $W=V$, it is also a maximal clique of $\Gint{s_e(C)}{e}$. Moreover $f_e(C)=e$ since the link of $uv$ ends at $e$, and Lemma~\ref{lem:charact} concludes.
\end{proof}

\begin{theorem}\label{thm:edge-correct}
Given a simple link stream $\cL$, Algorithm~\ref{alg:edge} outputs exactly all its maximal cliques, without duplication.
\end{theorem}

\begin{proof}
An output $(C,[\tau,e])$ satisfies $s_e(C)=\tau$, so it is a maximal clique by Lemma~\ref{lem:edge-cand}. Conversely, let $(C,[\tau,e])$ be a maximal clique, with canonical pair $uv$ and canonical link $\ell=(b,e,u,v)$. Observe that $C\setminus\{u,v\}\subseteq H_\ell$. By Lemma~\ref{lem:charact}, $\tau=s_e(C)$ is the start of a link between two vertices of $C$, and $\tau\ge b_e(uv)=b$, so $\tau\in\Theta_\ell$. The set $K=C\setminus\{u,v\}$ is a clique of $H^\tau_\ell$, and it is maximal in this graph because $C$ is a maximal clique of $\Gint{\tau}{e}$. Thus $C$ is examined as a candidate for $\ell$ and $\tau$, it passes both tests, and $(C,[\tau,e])$ is output. It is output only for its canonical link, for $\tau=s_e(C)$ and for $K=C\setminus\{u,v\}$, and the maximal cliques of $H^\tau_\ell$ are found without duplication, so it is output exactly once.
\end{proof}

\begin{theorem}\label{thm:degree-count}
A link stream with $m$ links in which every instantaneous graph has maximum degree at most $\Delta$ has $O(m\Delta^23^{\Delta/3})$ maximal cliques. Conversely, for all integers $p\ge 1$ and $R\ge 1$, there is a link stream with $\Delta=3p+1$ and $m=M_p+R$ links, where $M_p=(9p^2+3p)/2$, which has $(R+2)3^p$ maximal cliques. When $R\ge M_p$, this is $\Omega(m3^{\Delta/3})$.
\end{theorem}

\begin{proof}
By Theorem~\ref{thm:edge-correct}, every maximal clique is output by Algorithm~\ref{alg:edge} for some link $\ell$, some $\tau\in\Theta_\ell$ and some maximal clique of $H^\tau_\ell$. Observe that $H_\ell\cup\{u,v\}$ has at most $\Delta+1$ vertices, so that $|\Theta_\ell|=O(\Delta^2)$, and that $H^\tau_\ell$ has at most $\Delta-1$ vertices, so that it has at most $3^{(\Delta-1)/3}$ maximal cliques. This gives the upper bound.

For the lower bound, let $Q$ be the complete $p$-partite graph whose parts have three vertices each. Its maximal cliques are its $3^p$ transversals, that is, the sets containing one vertex of every part. Take two more vertices $a$ and $b$. Add the link $(0,2R+2,x,y)$ for every edge $xy$ of $Q$, and the links $(0,2R+2,a,v)$ and $(0,2R+2,b,v)$ for every vertex $v$ of $Q$. Finally, add the links $(2j,2j+1,a,b)$ for $j\in[R]$. There are $3p(3p-3)/2+6p=M_p$ links of the first kinds and $R$ links of $ab$. In every instantaneous graph, $a$ and $b$ have degree at most $3p+1$ and the vertices of $Q$ have degree $3p-1$, so that $\Delta=3p+1$. For every transversal $X$ of $Q$ and every $j\in[R]$, the set $C=X\cup\{a,b\}$ satisfies the three conditions of Lemma~\ref{lem:charact} for $\tau=2j$ and $e=2j+1$. Indeed, $s_e(C)=2j$ and $f_e(C)=2j+1$ are given by the link of $ab$, and a vertex of $Q$ outside $C$ belongs to the same part as a vertex of $X$, and is thus not adjacent to it. This gives $R3^p$ maximal cliques. There are no other maximal cliques ending at the times $2j+1$, since by Lemma~\ref{lem:charact} such a clique contains the ending pair $ab$ and starts at $2j$. The maximal cliques ending at $2R+2$ do not contain both $a$ and $b$, and they are the cliques $(X\cup\{a\},[0,2R+2])$ and $(X\cup\{b\},[0,2R+2])$. When $R\ge M_p$, we have $m=O(R)$, and thus $\alpha\ge R3^p=\Omega(m3^{\Delta/3})$.
\end{proof}

The proof gives the bound $\alpha\le m\binom{\Delta+1}{2}3^{(\Delta-1)/3}$. When $R\ge M_p$, the construction satisfies $m\le 2R$, and thus $\alpha\ge\frac m2 3^{(\Delta-1)/3}$. The gap between the two bounds is therefore a factor $O(\Delta^2)$. We now analyse the time complexity. In the rest of the paper, $\lambda(h)$ denotes the time delay, including the time to detect the end of the enumeration, of an algorithm listing the maximal cliques of an $h$-order graph given by its adjacency matrix. We have $\lambda(h)=O(h^3)$ by \cite{TIAS77}, and $\lambda(h)=O(h^{\omega})$ after a preprocessing in time $O(h^2)$ by \cite{MU04}, where $\omega$ is the exponent of matrix multiplication. We assume that $\lambda$ is non-decreasing and that $\lambda(O(h))=O(\lambda(h))$, which holds for these polynomial bounds.

\begin{lemma}\label{lem:edge-charging}
Send every candidate $C$ of Algorithm~\ref{alg:edge}, examined for a link $\ell=(b,e,u,v)$ and a threshold $\tau$, to the maximal clique $(C,[s_e(C),e])$. Then every candidate is sent to a maximal clique which is output when the candidate is examined or before, and at most $O(q^2\Delta^2)$ candidates are sent to each maximal clique.
\end{lemma}

\begin{proof}
By Lemma~\ref{lem:edge-cand}, $(C,[s_e(C),e])$ is a maximal clique. Let $\ell^*$ be its canonical link and $x^*y^*$ its canonical pair. The pair $uv$ is an ending pair of $C$ at $e$, so $x^*y^*$ is not larger than $uv$, and $\ell^*$ is either equal to $\ell$ or considered before $\ell$. If $\ell^*=\ell$, then the maximal clique is output for $\ell$ and for the threshold $s_e(C)\le\tau$. This threshold is considered before $\tau$, or is equal to $\tau$, in which case the candidate is the output itself. For the second claim, observe that a candidate sent to $(C,[s,e])$ is determined by the pair $uv$, which is one of the at most $\binom q2$ pairs of $C$, and by $\tau\in\Theta_\ell$. Thus there are $O(q^2\Delta^2)$ of them.
\end{proof}

\begin{theorem}\label{thm:degree-enum}
Given a simple link stream, Algorithm~\ref{alg:edge} can be implemented with setup time $O(m\log m)$ and amortized time delay $O(q^2\Delta^2(\Delta\log m+\Delta^2+\lambda(\Delta)))$. It can also be implemented with the same setup time and polynomial time delay $O(\Delta^4(\Delta\log m+\Delta^2+\lambda(\Delta)))$.
\end{theorem}

\begin{proof}
In the setup phase, we sort the links by end time and then by pair, and for every vertex $x$ we build an interval tree on the links incident to $x$. Such a tree returns the links of $x$ containing a time $t$ in time $O(\log m+h)$, where $h$ is the number of these links \cite{BKOS08}. We also compute $\Delta$ by sweeping the time and maintaining the degrees of the current instantaneous graph. This takes $O(m\log m)$ time in total.

When a link $\ell=(b,e,u,v)$ is considered, we query the trees of $u$ and $v$ to obtain $N_{G_e}(u)$ and $N_{G_e}(v)$, together with the start times of the corresponding links, and we compute $H_\ell$ in time $O(\log m+\Delta)$. Then we query the tree of every vertex of $H_\ell$ to obtain the edges of $G_e$ between vertices of $H_\ell\cup\{u,v\}$ together with their values $b_e$, in time $O(\Delta(\log m+\Delta))$. We put the values of $\Theta_\ell$ in a binary heap in time $O(\Delta^2)$, so that every threshold is extracted in increasing order in time $O(\log\Delta)$. For every $\tau\in\Theta_\ell$, we build the adjacency matrix of $H^\tau_\ell$ in time $O(\Delta^2)$ and list its maximal cliques with time delay $\lambda(\Delta)$. The two tests of the algorithm take time $O(q^2)$ per candidate.

Observe that every link and every threshold yield at least one candidate, since $b\in\Theta_\ell$ and since every graph, including the graph without vertex, has a maximal clique. Thus condition (i) of Lemma~\ref{lem:accounting} holds with $\kappa=O(\Delta\log m+\Delta^2+\lambda(\Delta))$, and conditions (ii) and (iii) hold with $\mu=O(q^2\Delta^2)$ by Lemma~\ref{lem:edge-charging}. Therefore, the algorithm has amortized time delay $O(q^2\Delta^2\kappa)$. The value of $q$ is not known after the setup phase, but $q\le\Delta+1$, so the amortized time delay is also $p=O(\Delta^4\kappa)$, which can be computed from $\Delta$ and $m$. Since a solution has at most $q\le\Delta+1$ vertices, Lemma~\ref{lem:buffer} gives the time delay $O(p)$.
\end{proof}

Observe that, when Algorithm~\ref{alg:edge} considers a link ending at $e$, it only uses the links that start at $e$ or before. It can therefore be run online, sweeping the time with the three phases described in Section~\ref{sec:dynamic}: at a time $t$, the links starting at $t$ are inserted, then the links ending at $t$ are considered, and finally they are deleted. We store the current neighbors of every vertex, with the start times of the corresponding links, in a balanced search tree, and we compute the sets $H_\ell$ and the edges between their vertices with an array that marks the vertices of $H_\ell\cup\{u,v\}$. In that case, every maximal clique is output at its end time, and the total time is $O(m\log m)+\alpha O(q^2\Delta^2(\Delta^2+\lambda(\Delta)))$. Note also that the algorithm of Baudin et al.\ \cite{BMT24} anchors every clique at a link starting at its start time and lists all the cliques of the instantaneous graph containing this link, which may cost a factor $2^q$ per maximal clique. Algorithm~\ref{alg:edge} only lists maximal cliques of the graphs $H^\tau_\ell$, one threshold at a time, which replaces this factor by $O(q^2\Delta^2)$.

\section{Algorithms parameterized by the degeneracy}\label{sec:degeneracy}

The bounds of Section~\ref{sec:degree} depend on the maximum degree, which can be much larger than the degeneracy. We first observe that the number of maximal cliques cannot be bounded by the number of links times a function of the degeneracy.

\begin{example}\label{ex:quadratic}
Let $N\ge 2$. Take the vertices $a$, $b$ and $c_1,\dots,c_N$, the links $(0,2N,a,c_i)$ and $(0,2N,b,c_i)$ for $i\in[N]$, and the links $(2j-1,2j-1,a,b)$ for $j\in[N]$. Every instantaneous graph is $K_{2,N}$, possibly with the edge $ab$, so $k=2$, and $m=3N$. By Lemma~\ref{lem:charact}, $(\{a,b,c_i\},[2j-1,2j-1])$ is a maximal clique for all $i,j\in[N]$, so that $\alpha\ge N^2=m^2/9$.
\end{example}

The relevant quantity is therefore not $m$ but the product of the number $n$ of vertices and of the number $|\cT|$ of end times, as in the bounds of \cite{HMNS17,BHM+19}.

\subsection{Orders, anchors and suffixes}

\begin{definition}\label{def:anchor}
Assume that for every end time $e\in\cT$ we are given a total order $\sigma_e$ on $V$. For a vertex $w$, $\Np_e(w)$ is the set of neighbors of $w$ in $G_e$ that have higher rank than $w$ in $\sigma_e$, and $\Np_e[w]=\Np_e(w)\cup\{w\}$. For a set $C$ of vertices, $\first_e(C)$ is the vertex of $C$ with lowest rank in $\sigma_e$, and for $w\in C$, $C_{\ge w}$ is the set of vertices of $C$ with rank at least $\sigma_e(w)$. For a pair $xy\in P_e$, the anchors of $xy$ at $e$ are the vertices $w$ such that $x,y\in\Np_e[w]$.
\end{definition}

We denote by $c$ the largest size of a set $\Np_e(w)$. If $\sigma_e$ is a degeneracy ordering of $G_e$ for every $e\in\cT$, then $c\le k$. Observe that every ending pair $xy\in P_e$ has at least one anchor, namely the vertex of lowest rank among $x$ and $y$, since $xy$ is an edge of $G_e$. Ordering the vertices of a clique by increasing rank in $\sigma_e$, we will see cliques as strings; the set $C_{\ge w}$ then corresponds to a suffix of the string of $C$, which is proper if $w\ne\first_e(C)$. For an anchor $w$ of $xy$, we call $W=\Np_e[w]$ its window, we let $Z=\{w,x,y\}$, which is a clique of $G_e$, we let $\Theta_w$ be the set of values $b_e(x'y')$ over the edges $x'y'$ of $G_e$ with $x',y'\in W$, and we let $F^\tau_w=\Gint{\tau}{e}[W]$.

\begin{lemma}\label{lem:outwindow}
Let $C$ be a clique of $G_e$ and $w\in C$. Then $C_{\ge w}\subseteq\Np_e[w]$. In particular, $C\subseteq\Np_e[\first_e(C)]$.
\end{lemma}

\begin{proof}
Every vertex of $C_{\ge w}$ distinct from $w$ is adjacent to $w$ in $G_e$ and has higher rank than $w$ in $\sigma_e$.
\end{proof}

The next lemma is the counterpart, for link streams, of Corollary~2 of \cite{Man19}. It is illustrated in Figure~\ref{fig:suffix}.

\begin{lemma}\label{lem:suffix}
Let $e\in\cT$, $w\in V$ and $\tau\le e$, and let $C$ be a maximal clique of $F^\tau_w$ such that $w\in C$, $s_e(C)=\tau$ and $f_e(C)=e$. If $C$ is not a maximal clique of $\Gint{\tau}{e}$, then there is a maximal clique $(C^*,[\tau,e])$ of $\cL$ such that $\first_e(C^*)$ has lower rank than $w$ in $\sigma_e$ and $C=C^*_{\ge w}$. In particular, the string of $C$ is a proper suffix of the string of $C^*$.
\end{lemma}

\begin{proof}
Let $z\notin C$ be adjacent in $\Gint{\tau}{e}$ to all the vertices of $C$. Assume by contradiction that $z$ has higher rank than $w$ in $\sigma_e$. Since $z$ is adjacent to $w$ in $G_e$, this implies that $z\in\Np_e(w)$, and thus that $C\cup\{z\}$ is a clique of $F^\tau_w$, which contradicts the maximality of $C$. Therefore $z$ has lower rank than $w$.

Let $C^*$ be a maximal clique of $\Gint{\tau}{e}$ containing $C\cup\{z\}$. Every pair of $C^*$ is an edge of $\Gint{\tau}{e}$, so $s_e(C^*)\le\tau$, and $s_e(C^*)\ge s_e(C)=\tau$ since $C\subseteq C^*$. Thus $s_e(C^*)=\tau$. Similarly, $f_e(C^*)\ge e$ since every pair of $C^*$ is an edge of $G_e$, and $f_e(C^*)\le f_e(C)=e$, so that $f_e(C^*)=e$. By Lemma~\ref{lem:charact}, $(C^*,[\tau,e])$ is a maximal clique of $\cL$. Observe that $\first_e(C^*)$ has rank at most $\sigma_e(z)<\sigma_e(w)$.

Finally, $C^*_{\ge w}\subseteq\Np_e[w]=W$ by Lemma~\ref{lem:outwindow}, and $C^*_{\ge w}$ is a clique of $\Gint{\tau}{e}$, hence of $F^\tau_w$. It contains $C$, since every vertex of $C$ belongs to $W$ and has thus rank at least $\sigma_e(w)$. By maximality of $C$ in $F^\tau_w$, this implies that $C^*_{\ge w}=C$.
\end{proof}

\begin{figure}[t]
\centering
\begin{tikzpicture}[x=1.15cm,every node/.style={font=\small}]
  \draw[->,gray] (-0.3,0) -- (9.6,0) node[right] {$\sigma_e$};
  \foreach \x/\l in {0.5/,1.5/z,2.5/,3.5/w,4.5/x,5.5/,6.5/y,7.5/,8.5/} {
    \fill (\x,0) circle (2pt);
    \node[below=3pt] at (\x,0) {$\l$};
  }
  \node[below=3pt] at (0.5,0) {$w^*$};
  \draw[blue,thick,rounded corners] (0.2,0.3) rectangle (8.8,1.05);
  \node[blue] at (4.5,0.68) {$C^*$: a maximal clique of $\Gint{\tau}{e}$, output at its first vertex $w^*$};
  \draw[decorate,decoration={brace,amplitude=5pt}] (3.3,1.25) -- (8.7,1.25);
  \node[above] at (6,1.42) {window $\Np_e[w]$};
  \draw[red,thick,rounded corners] (3.2,-0.85) rectangle (8.8,-0.45);
  \node[red,below] at (6,-0.85) {$C=C^*_{\ge w}$: a candidate at the anchor $w$, rejected since $(\tau,C)\in\cD$};
\end{tikzpicture}
\caption{Lemma~\ref{lem:suffix}. A vertex $z$ extending a candidate $C$ found in the window of $w$ has lower rank than $w$. The maximal clique $C^*$ containing $C\cup\{z\}$ is output at an earlier anchor, and $C$ is its suffix starting at $w$.}\label{fig:suffix}
\end{figure}
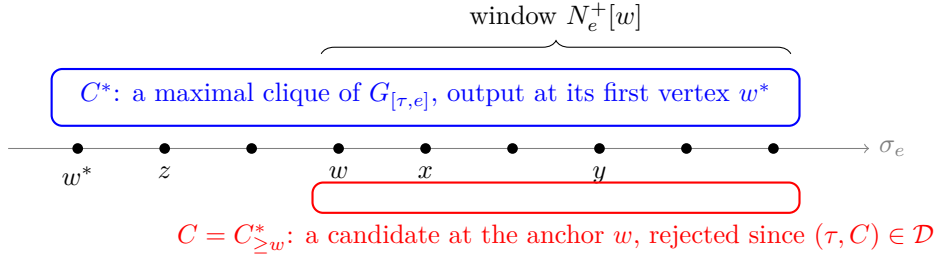

\subsection{The algorithm for given orders}

The outline of the second algorithm is the following. We consider the end times in increasing order. For every end time $e$, we consider the ending pairs in increasing order, the anchors of each pair by increasing rank in $\sigma_e$, and the thresholds in increasing order. For each of them, we find all the maximal cliques of $F^\tau_w$ which contain $Z$. Such a candidate is output if it passes three tests. The first two are those of Algorithm~\ref{alg:edge}. For the third one, we use a dictionary $\cD$ of pairs $(\tau,S)$, emptied at every end time, in which we insert the proper suffixes of the maximal cliques that are output, and a candidate is rejected if it is found in $\cD$. The procedure is described in Algorithm~\ref{alg:vertex}.

\begin{algorithm}[t]
\caption{Maximal cliques anchored at vertices}\label{alg:vertex}
\KwData{A simple link stream $\cL$ and an order $\sigma_e$ for every $e\in\cT$.}
\KwResult{All the maximal cliques of $\cL$.}
\For{every $e\in\cT$ in increasing order}{
  Initialize $\cD$ an empty dictionary\;
  \For{every pair $xy\in P_e$ in increasing order}{
    \For{every anchor $w$ of $xy$ at $e$, by increasing rank in $\sigma_e$}{
      $W\gets\Np_e[w]$, $Z\gets\{w,x,y\}$\;
      \For{every $\tau\in\Theta_w$ with $\tau\ge s_e(Z)$, in increasing order}{
        \For{every maximal clique $C$ of $F^\tau_w$ which contains $Z$}{
          \uIf{$s_e(C)<\tau$ or $xy$ is not the smallest pair of $C$ in $P_e$ or $(\tau,C)\in\cD$}{Reject it.}
          \Else{Insert $(\tau,C_{\ge v})$ in $\cD$ for every $v\in C\setminus\{w\}$\;
            Output $(C,[\tau,e])$.}
        }
      }
    }
  }
}
\end{algorithm}

The maximal cliques of $F^\tau_w$ which contain $Z$ are the sets $Z\cup K$ where $K$ is a maximal clique of the subgraph of $F^\tau_w$ induced by the common neighbors of the vertices of $Z$; this subgraph has at most $c$ vertices. Observe that $Z$ is a clique of $F^\tau_w$ when $\tau\ge s_e(Z)$. In the following, a candidate is a tuple $(e,xy,w,\tau,C)$ examined by the innermost loop.

\begin{lemma}\label{lem:dict}
Whenever a candidate is examined, for every element $(\tau,S)$ of $\cD$, $S$ is a proper subset of a clique $C$ of $\Gint{\tau}{e}$ such that $(C,[\tau,e])$ has been output, where $e$ is the current end time. In particular, $S$ is not a maximal clique of $\Gint{\tau}{e}$.
\end{lemma}

\begin{proof}
An element $(\tau,C_{\ge v})$ is inserted when $(C,[\tau,e])$ is output, and $C$ is then a clique of $\Gint{\tau}{e}$. Since $v\ne w=\first_e(C)$, the set $C_{\ge v}$ does not contain $w$ and is a proper subset of $C$. Since $C$ is a clique of $\Gint{\tau}{e}$, $C_{\ge v}$ is not a maximal clique of this graph.
\end{proof}

\begin{lemma}\label{lem:vertex-complete}
Let $(C,[\tau,e])$ be a maximal clique of $\cL$, let $xy$ be its canonical pair and let $w=\first_e(C)$. Then Algorithm~\ref{alg:vertex} outputs $(C,[\tau,e])$ when it considers the pair $xy$, the anchor $w$ and the threshold $\tau$, and does not output it at any other moment.
\end{lemma}

\begin{proof}
By Lemma~\ref{lem:outwindow}, $C\subseteq\Np_e[w]$, so $w$ is an anchor of $xy$ and $Z\subseteq C$. By Lemma~\ref{lem:charact}, $\tau=s_e(C)$. This implies that $\tau\in\Theta_w$ and $\tau\ge s_e(Z)$. Lemma~\ref{lem:charact} also implies that $C$ is a maximal clique of $\Gint{\tau}{e}$, hence of $F^\tau_w$. Thus $(e,xy,w,\tau,C)$ is a candidate. It satisfies $s_e(C)=\tau$, its pair $xy$ is canonical, and $(\tau,C)\notin\cD$ by Lemma~\ref{lem:dict}. Therefore it is output.

Conversely, if $(C,[\tau,e])$ is output for the pair $xy$ and the anchor $w$, then $w\in C\subseteq\Np_e[w]$, so $w=\first_e(C)$. The second test forces $xy$ to be the canonical pair, and the threshold is $\tau$. Since the maximal cliques of $F^\tau_w$ are found without duplication, $(C,[\tau,e])$ is output only once.
\end{proof}

\begin{theorem}\label{thm:vertex-correct}
Given a simple link stream $\cL$ and an order $\sigma_e$ for every $e\in\cT$, Algorithm~\ref{alg:vertex} outputs exactly all the maximal cliques of $\cL$, without duplication.
\end{theorem}

\begin{proof}
By Lemma~\ref{lem:vertex-complete}, every maximal clique is output exactly once. It remains to prove that every output is a maximal clique. Let $(C,[\tau,e])$ be output for the pair $xy$ and the anchor $w$. Then $C$ is a maximal clique of $F^\tau_w$ containing $w$, $s_e(C)=\tau$, and $f_e(C)=e$ since $xy\subseteq C$. By Lemma~\ref{lem:charact}, it suffices to show that $C$ is a maximal clique of $\Gint{\tau}{e}$. Assume by contradiction that it is not. By Lemma~\ref{lem:suffix}, there is a maximal clique $(C^*,[\tau,e])$ with $C=C^*_{\ge w}$ such that $w^*=\first_e(C^*)$ has lower rank than $w$. Let $x^*y^*$ be the canonical pair of $C^*$. Since $xy\subseteq C\subseteq C^*$, the pair $x^*y^*$ is not larger than $xy$. By Lemma~\ref{lem:vertex-complete}, $(C^*,[\tau,e])$ is output when the pair is $x^*y^*$ and the anchor is $w^*$. If $x^*y^*$ is smaller than $xy$, this happens before the pair $xy$ is considered. Otherwise $x^*y^*=xy$, and $w^*$ is an anchor of $xy$ considered before $w$. In both cases, since $w\in C^*\setminus\{w^*\}$, the element $(\tau,C^*_{\ge w})=(\tau,C)$ was inserted in $\cD$ before $C$ was examined, and $C$ was rejected, a contradiction.
\end{proof}

\begin{lemma}\label{lem:vertex-charging}
For every candidate $(e,xy,w,\tau,C)$, let $s=s_e(C)$. If $C$ is a maximal clique of $\Gint{s}{e}$, send the candidate to $(C,[s,e])$. Otherwise, send it to a maximal clique $(C^*,[s,e])$ given by Lemma~\ref{lem:suffix}, chosen once and for all if there are several. Then every candidate is sent to a maximal clique which is output when the candidate is examined or before, and at most $q^3\binom{c+1}{2}$ candidates are sent to each maximal clique.
\end{lemma}

\begin{proof}
By Lemma~\ref{lem:shadow}, $C$ is a maximal clique of $F^s_w$ with $s_e(C)=s$, and $f_e(C)=e$. Thus the image is well defined. It is $(C,[s,e])$, which is then a maximal clique by Lemma~\ref{lem:charact}, or $(C^*,[s,e])$ with $C=C^*_{\ge w}$ and $\first_e(C^*)$ of lower rank than $w$. In both cases, the image is $(C^*,[s,e])$ for a set $C^*$ with $C=C^*_{\ge w}$, where $C^*=C$ in the first case.

Let $w^*=\first_e(C^*)$ and let $x^*y^*$ be the canonical pair of $C^*$. By Lemma~\ref{lem:vertex-complete}, the image is output for the pair $x^*y^*$, the anchor $w^*$ and the threshold $s$. Since $xy\subseteq C^*$, the pair $x^*y^*$ is not larger than $xy$. If it is smaller, the image is output before the candidate is examined. Assume now that $x^*y^*=xy$. If $w^*$ has lower rank than $w$, the image is output while an earlier anchor of $xy$ is considered. Otherwise $w^*=w$, which implies that $C^*=C^*_{\ge w}=C$, and the image is output at the threshold $s\le\tau$. If $s<\tau$ this happens before. If $s=\tau$, the candidate passes the first two tests and, by Lemma~\ref{lem:dict}, the third one, so that it is the output itself.

For the second claim, observe that a candidate sent to $(C^*,[s,e])$ is determined by the vertex $w\in C^*$, which determines $C=C^*_{\ge w}$, by the pair $xy$, which is one of the pairs of $C$, and by $\tau\in\Theta_w$, which has at most $\binom{c+1}{2}$ elements.
\end{proof}

\begin{theorem}\label{thm:vertex-given}
Assume that, for every $e\in\cT$ and every pair $xy\in P_e$, the anchors of $xy$ at $e$ can be listed by increasing rank in $\sigma_e$ with time delay $A$, that $\Np_e(w)$ can be computed and sorted by rank in $\sigma_e$ in time $O(cA)$, and that the values $b_e$ of the edges inside $\Np_e[w]$ can be computed in time $O(c^2A)$. Then Algorithm~\ref{alg:vertex} has amortized time delay $O(q^3c^2(c^2A+\lambda(c)+q\log n))$.
\end{theorem}

\begin{proof}
Observe first that every ending pair has an anchor, that every anchor yields the threshold $s_e(Z)\in\Theta_w$, and that every threshold yields at least one candidate, since $Z$ is contained in a maximal clique of $F^\tau_w$. Therefore, between two consecutive candidates, the algorithm moves to the next end time, pair, anchor or threshold, and computes a window and the values $b_e$ inside it in time $O(c^2A)$. It puts the values of $\Theta_w$ in a binary heap in time $O(c^2)$ and extracts the next one in time $O(\log c)$. Then it builds the adjacency matrix of $F^\tau_w$ in time $O(c^2)$, and finds the next maximal clique in time $\lambda(c)$. Sorting the vertices of a candidate by rank and the first two tests take time $O(q^2)=O(c^2A)$, since $q\le c+1$. We store $\cD$ as a balanced search tree indexed by $\tau$, whose entries are tries, every node keeping its children in a balanced search tree. Let $C=(v_1,\dots,v_r)$ be an output, ordered by increasing rank. The strings of its proper suffixes $C_{\ge v_i}$, $2\le i\le r$, read backwards, are exactly the prefixes of the string $(v_r,\dots,v_2)$. We therefore insert this single string in the trie of $\tau$, every node of the trie representing an element of $\cD$, and a candidate is in $\cD$ if and only if its string read backwards is the label of a path from the root of the trie of its threshold. A search and the insertions of an output then take time $O(q\log n)$, since the thresholds at a time $e$ are start times of links present at $e$, and they add $O(q)$ nodes. At a new end time, a new empty dictionary is created in constant time, and the previous one remains allocated, so that the space used by $\cD$ is proportional to the total size of the outputs. Lemma~\ref{lem:accounting} and Lemma~\ref{lem:vertex-charging} conclude.
\end{proof}

\subsection{The number of maximal cliques}

\begin{theorem}\label{thm:degeneracy-count}
A link stream with $n$ vertices, $|\cT|$ distinct end times and $k$-degenerate instantaneous graphs has $O(n|\cT|k^23^{k/3})$ maximal cliques. More precisely, it has $O(k^23^{k/3}\sum_{e\in\cT}n_e)$ maximal cliques, where $n_e$ is the number of non-isolated vertices of $G_e$. Conversely, for all integers $p,R\ge 1$ and $h\ge 3$, there is a link stream with $k=3p+2$, $n=3p+2+h$ and $|\cT|=R+1$ which has $(R+2)h3^p$ maximal cliques. When $h\ge k$, this is $\Omega(n|\cT|3^{k/3})$.
\end{theorem}

\begin{proof}
For the upper bound, choose for every $e\in\cT$ a degeneracy ordering $\sigma_e$ of $G_e$, so that $c\le k$. By Lemma~\ref{lem:vertex-complete}, every maximal clique $(C,[\tau,e])$ is determined by $e$, by $w=\first_e(C)$, by $\tau\in\Theta_w$, and by $C\setminus\{w\}$, which is a maximal clique of the subgraph of $F^\tau_w$ induced by the neighbors of $w$ in $F^\tau_w$. This subgraph has at most $k$ vertices, so it has at most $3^{k/3}$ maximal cliques, and $|\Theta_w|\le\binom{k+1}{2}$. The vertex $w$ is an anchor, hence a non-isolated vertex of $G_e$. Summing over the end times $e$ and the $n_e\le n$ non-isolated vertices of $G_e$ gives the bounds.

For the lower bound, take the link stream of the proof of Theorem~\ref{thm:degree-count}, and add an independent set $B$ of $h$ vertices. For every $z\in B$, add the links $(0,2R+2,z,v)$ for every vertex $v$ of $Q$ and for $v\in\{a,b\}$. The end times are $3,5,\dots,2R+1$ and $2R+2$. In every instantaneous graph, a vertex of $B$ has degree $3p+2$ and the graph obtained by removing $B$ has $3p+2$ vertices, so that the degeneracy is at most $3p+2$. When $ab$ is present, the vertices of $Q$ have degree $h+3p-1\ge 3p+2$, and $a$ and $b$ have degree $h+3p+1$, so that the minimum degree is $3p+2$ and $k=3p+2$. Let $j\in[R]$, $z\in B$ and let $X$ be a transversal of $Q$. The set $C=X\cup\{a,b,z\}$ satisfies the conditions of Lemma~\ref{lem:charact} for $\tau=2j$ and $e=2j+1$. Indeed, a vertex outside $C$ either belongs to $B$, and is then not adjacent to $z$, or belongs to the same part of $Q$ as a vertex of $X$, and is then not adjacent to it. This gives $Rh3^p$ maximal cliques. As in the proof of Theorem~\ref{thm:degree-count}, the other maximal cliques are the cliques $(X\cup\{a,z\},[0,2R+2])$ and $(X\cup\{b,z\},[0,2R+2])$. When $h\ge k$, we have $n=O(h)$ and $|\cT|=O(R)$, so that $\alpha\ge Rh3^p=\Omega(n|\cT|3^{k/3})$.
\end{proof}

The proof also gives the more local bound $\alpha\le\sum_{e\in\cT}\sum_{w\in A_e}|\Theta_w|3^{|\Np_e(w)|/3}$, where $A_e$ is the set of anchors at $e$ for a degeneracy ordering $\sigma_e$ of $G_e$, and $\Theta_w$ is taken at time $e$. When $h\ge k$, the construction satisfies $n\le 2h$ and $|\cT|\le 2R$, and thus $\alpha\ge\frac{n|\cT|}{4}3^{(k-2)/3}$. The gap between the two bounds of Theorem~\ref{thm:degeneracy-count} is therefore a factor $O(k^2)$. Observe that Theorems~\ref{thm:degree-count} and~\ref{thm:degeneracy-count} are incomparable. In Example~\ref{ex:quadratic}, $\sum_{e\in\cT}n_e=\Theta(m^2)$ and $\Delta=N+1$. Conversely, take $N$ disjoint pairs which interact during $[0,2R+1]$, and $R$ links of another pair which end at distinct times. Then $\Delta=1$ and $m=N+R$, while $\sum_{e\in\cT}n_e=\Theta(NR)$.

\subsection{Maintaining the orders over time}\label{sec:dynamic}

It remains to provide orders $\sigma_e$ with small out-degree, and the operations required by Theorem~\ref{thm:vertex-given}, without computing a degeneracy ordering at every end time, which could cost $\Theta(m)$ time per end time. We sweep the time and maintain the instantaneous graph under insertions and deletions of edges. To respect closed intervals, the events at a time $t$ are processed in three phases. First the links starting at $t$ are inserted, then the end time $t$ is processed if $t\in\cT$, and finally the links ending at $t$ are deleted. Every link causes one insertion and one deletion. We use the following data structure.

\begin{lemma}[\cite{BHNT15}, {\cite[Section 2 and Lemma 7]{HNW20}}, see also {\cite[Corollary 3.3]{LSY+22}}]\label{lem:lds}
There is a data structure which maintains, under insertions and deletions of edges in a graph on $n$ vertices, a level for every vertex such that the following holds. Orienting every edge from the lower to the higher level, and from the smaller to the larger index between vertices of the same level, gives an acyclic orientation in which every out-degree is $O(a)$, where $a$ is the arboricity of the current graph. The data structure is deterministic, and the amortized time, the amortized number of level changes and the amortized number of reoriented edges per update are $O(\log^2n)$.
\end{lemma}

In \cite{HNW20}, the edges between two vertices of the same level are oriented arbitrarily, and every vertex has $O(a)$ neighbors at its level or at higher levels. Thus orienting them by index keeps the out-degrees in $O(a)$ and makes the orientation acyclic. When a vertex changes level, the edges whose orientation changes join it to neighbors whose level is at least the smaller of its old and new levels. The data structure of \cite{HNW20} scans these neighbors during the move, so that these reorientations are paid for by the bounds of \cite{HNW20}. Liu et al.\ state this acyclic orientation explicitly for their parallel version \cite{LSY+22}. The out-degree bound holds between two updates, and may be violated while the levels are being updated. Since the arboricity is at most the degeneracy, the out-degrees are $O(k)$. We let $\sigma_e$ be the order of the vertices by level at time $e$ and then by index. Then the orientation of $G_e$ given by Lemma~\ref{lem:lds} is the orientation along $\sigma_e$, and $c=O(k)$. In addition, we maintain for every pair $xy$ the set $\Pi(xy)$ of the vertices $w$ such that $x,y\in\Np(w)$, where $\Np(w)$ denotes the current out-neighborhood of $w$. The anchors of an ending pair $xy$ at $e$ are the vertex of lowest rank among $x$ and $y$ and the vertices of $\Pi(xy)$.

\begin{theorem}\label{thm:online}
The maximal cliques of a simple link stream can be enumerated online, every maximal clique being output at its end time, in total time $O(m\log m+mk\log^3n)+\alpha O(k^5(k^2\log n+\lambda(k)))$.
\end{theorem}

\begin{proof}
We sort the events in time $O(m\log m)$, which is not needed if they are given in chronological order. We maintain the data structure of Lemma~\ref{lem:lds}, the current links in a balanced search tree indexed by pairs, and the nonempty sets $\Pi(xy)$ in a balanced search tree indexed by pairs, each of them stored as a balanced search tree. We update the sets $\Pi(xy)$ once per insertion or deletion of an edge, after the data structure of Lemma~\ref{lem:lds} has been updated, so that the out-degrees are $O(k)$ before and after the update. For every vertex $w$ whose out-neighborhood has changed, we compare its old and new out-neighborhoods. For every vertex $u$ which has left $\Np(w)$, we remove $w$ from the sets $\Pi(ux)$ for $x$ in the old out-neighborhood of $w$, and for every vertex $u$ which has entered $\Np(w)$, we add $w$ to the sets $\Pi(ux)$ for $x$ in the new out-neighborhood of $w$. This takes time $O(k\log n)$ per change, since the out-degrees are $O(k)$ before and after the update. The number of such changes is at most the number of reoriented, inserted and deleted edges. By Lemma~\ref{lem:lds}, the maintenance takes total time $O(mk\log^3n)$.

At an end time $e$, for every ending pair $xy$, we sort its anchors by level and index, in time $O(h\log n)$ for $h$ anchors. Every anchor yields at least one candidate. The window of an anchor is its current out-neighborhood, and the values $b_e$ inside it are read in the tree of current links in time $O(c^2\log n)$. The rest of the work is accounted for as in the proof of Theorem~\ref{thm:vertex-given}, with a cost $O(c^2\log n)$ per anchor instead of $O(c^2A)$. By Lemma~\ref{lem:vertex-charging}, there are $O(\alpha q^3c^2)$ candidates, with $q\le c+1=O(k)$. Summing gives the bound. Observe finally that every maximal clique ending at $e$ is output while $e$ is processed.
\end{proof}

The sweep does not give a bound on the time delay, since many links may be inserted before the first end time, and this work cannot be charged to the outputs that come before it. We therefore perform the sweep during the setup phase, record the history of the data structure, and answer the queries of every end time from this history.

\begin{theorem}\label{thm:delay}
Given a simple link stream, Algorithm~\ref{alg:vertex} can be implemented with polynomial time delay $O(k^5(k^2\log m+\lambda(k)))$, after a setup phase in time $O(m\log m+mk^2\log^3n)$. It requires $O(mk^2\log^3n)$ space, besides the space used by the dictionary $\cD$ and by Lemma~\ref{lem:buffer}.
\end{theorem}

\begin{proof}
During the setup phase, we run the sweep of Theorem~\ref{thm:online} without enumerating, and we make the following structures persistent with the path-copying technique of Driscoll et al.\ \cite{DSST89}. For every tree, we record the list of its versions together with the time and the phase of the event that created them. Every set $\Pi(xy)$ is stored as a balanced search tree ordered by level and then by index, and the out-neighborhood $\Np(w)$ of every vertex is stored as a balanced search tree ordered by index, so that a change of level of a vertex does not modify the out-neighborhoods that contain it. As in the proof of Theorem~\ref{thm:online}, these trees are updated once per insertion or deletion of an edge, after the data structure of Lemma~\ref{lem:lds} has been updated. An update of a persistent balanced search tree takes time and space $O(\log n)$. The keys of the trees $\Pi(xy)$ are stored values: a vertex $w$ is inserted with the pair formed by its level at that moment and its index, so that the old versions are never modified. When the level or the out-neighborhood of $w$ changes, we remove the entries of $w$, with their old keys, from the trees $\Pi(xy)$ with $x,y$ in the old out-neighborhood of $w$, and we insert them, with their new keys, in the trees $\Pi(xy)$ with $x,y$ in the new out-neighborhood of $w$. Since the out-degrees are $O(k)$ before and after every update, this costs $O(k^2)$ updates of persistent trees per level change and per reoriented, inserted or deleted edge. By Lemma~\ref{lem:lds}, the number of level changes and of reoriented edges is $O(m\log^2n)$. Thus the setup phase performs $O(mk^2\log^2n)$ updates of persistent trees, in time and space $O(mk^2\log^3n)$. We also record, for every vertex, the sequence of its levels with the times of the changes, and the largest out-degree $c$ reached when an end time is processed. The bound $O(k)$ on the out-degrees, which holds after every update, is used for the maintenance, and $c$ is only used to set the delay. Finally, we sort the links by end time and by pair, and we store, for every pair, its links sorted by start time, in time $O(m\log m)$.

During the enumeration phase, for an end time $e$ and an ending pair $xy$, we find by binary search the versions current at $e$, that is, the versions reached after the insertions at $e$ and before the deletions at $e$, in time $O(\log m)$, using $\log n=O(\log m)$. We list the anchors of $xy$ by an in-order traversal of the version of $\Pi(xy)$, merged with the vertex of lowest rank among $x$ and $y$, whose levels at time $e$ are found by binary search. This has time delay $O(\log m)$. For an anchor $w$, we read the version of $\Np(w)$ in time $O(\log m+c)$. We sort it by rank after finding the levels of its vertices at time $e$ by binary search, in time $O(c\log m)$. Then we find the value $b_e$ of every edge inside $\Np_e[w]$ by binary search in the links of the pair, in time $O(c^2\log m)$. Theorem~\ref{thm:vertex-given} with $A=O(\log m)$ gives an amortized time delay $O(q^3c^2(c^2\log m+\lambda(c)+q\log n))$. Since $q\le c+1$, this is also $p=O(c^5(c^2\log m+\lambda(c)))=O(k^5(k^2\log m+\lambda(k)))$, which is known at the end of the setup phase. Since every solution has at most $q\le c+1$ vertices, Lemma~\ref{lem:buffer} gives the time delay $O(p)$.
\end{proof}

We conclude this section with three observations. First, $\alpha$ does not depend on the orders used by the algorithm. Thus, by Theorem~\ref{thm:degeneracy-count}, the total running time of the algorithms of Theorems~\ref{thm:online} and~\ref{thm:delay} is $O(m\log m+mk^2\log^3n)+O(n|\cT|3^{k/3})\poly(k)\log m$, which matches the lower bound of Theorem~\ref{thm:degeneracy-count} up to a factor $\poly(k)\log m$ when the output dominates the setup phase. Second, a static graph $G$ can be seen as the link stream with a link $(0,1,u,v)$ for every edge $uv$, whose maximal cliques are the pairs $(C,[0,1])$ for the maximal cliques $C$ of $G$ with at least two vertices. With a degeneracy ordering, Algorithm~\ref{alg:vertex} then reduces to the algorithm of \cite{Man19}, with the generalized suffix tree replaced by the dictionary $\cD$, and Theorem~\ref{thm:vertex-given}, with the sets $\Pi(xy)$ computed in time $O(m+nk^2)$, gives an amortized time delay polynomial in the degeneracy, as in \cite{ManDelay}. Third, the dictionary $\cD$ and Lemma~\ref{lem:buffer} may require exponential space. As in Theorem~3 of \cite{Man19}, $\cD$ can be avoided by testing the maximality of a candidate directly, scanning the neighbors of the anchor in $G_e$ that have lower rank, in time $O(\Delta q)$. This reintroduces the maximum degree. Using the geometric amortization of \cite{CS23} instead of Lemma~\ref{lem:buffer}, this gives a time delay $\poly(k)\Delta\log m\log\alpha$ in space polynomial in the input, where $\log\alpha=O(\log n+\log|\cT|+k)$ by Theorem~\ref{thm:degeneracy-count}.

\section{Application to \texorpdfstring{$\Delta$}{Delta}-cliques}\label{sec:delta}

A temporal graph $\mathcal{G}=(V,\mathcal{E},[t_{\min},t_{\max}])$ is given by a set $\mathcal{E}$ of time-edges $(uv,t)$, with integer times $t\in[t_{\min},t_{\max}]$. As in the introduction, we denote by $\theta$ the integer that defines $\Delta$-cliques, and we assume that $0\le\theta\le t_{\max}-t_{\min}$, since otherwise there is no $\Delta$-clique. We use the discrete definition of \cite{HMNS17}, which differs from the original definition of \cite{VLM16} in that time is discrete. A $\Delta$-clique is a pair $(X,[a,b])$ with $|X|\ge 2$, integers $a$ and $b$ with $[a,b]\subseteq[t_{\min},t_{\max}]$ and $b-a\ge\theta$, such that for every integer $\tau\in[a,b-\theta]$ and all distinct $u,v\in X$, there is a time-edge $(uv,t)$ with $t\in[\tau,\tau+\theta]$. The $\Delta$-slice degeneracy of $\mathcal{G}$ is the largest degeneracy of the graphs $G^\theta_\tau$, for integers $\tau\in[t_{\min},t_{\max}]$, whose edges are the pairs having a time-edge in $[\tau,\tau+\theta]$ \cite{HMNS17}.

For every pair $uv$, the integers $\tau\in[t_{\min},t_{\max}-\theta]$ such that $uv$ has a time-edge in $[\tau,\tau+\theta]$ form a union of maximal sets of consecutive integers $\{\beta,\dots,\varepsilon\}$. We make each of them a link $(\beta,\varepsilon,u,v)$ of a link stream $\cL_\theta$. Every time-edge creates at most one such set and the sets are merged, so $\cL_\theta$ has at most $|\mathcal{E}|$ links, and it is simple.

\begin{lemma}\label{lem:delta}
The map $(X,[a,b])\mapsto(X,[a,b-\theta])$ is a bijection between the maximal $\Delta$-cliques of $\mathcal{G}$ and the maximal cliques of $\cL_\theta$. Moreover, the instantaneous graph of $\cL_\theta$ at an integer $\tau\in[t_{\min},t_{\max}-\theta]$ is $G^\theta_\tau$.
\end{lemma}

\begin{proof}
A pair $uv$ has a time-edge in $[\tau,\tau+\theta]$ for every integer $\tau\in[a,b-\theta]$ if and only if these integers belong to one maximal set of consecutive integers of $uv$, that is, if and only if a link of $uv$ in $\cL_\theta$ contains $[a,b-\theta]$. Thus $(X,[a,b])$ is a $\Delta$-clique if and only if $(X,[a,b-\theta])$ is a clique of $\cL_\theta$. The map preserves inclusion in both directions, so it sends the maximal $\Delta$-cliques to the cliques that are maximal among the cliques of $\cL_\theta$ with integer endpoints. Such a clique is maximal among all cliques. Indeed, if a clique $(C',[l,r])$ contained it, every link of a pair of $C'$ would contain $[l,r]$ and, having integer endpoints, also $[\lfloor l\rfloor,\lceil r\rceil]$, so that $(C',[\lfloor l\rfloor,\lceil r\rceil])$ would be a larger clique with integer endpoints. Conversely, the maximal cliques of $\cL_\theta$ have integer endpoints by Lemma~\ref{lem:charact}. The last claim is the definition of $G^\theta_\tau$.
\end{proof}

For an integer $\tau>t_{\max}-\theta$, the graph $G^\theta_\tau$ is a subgraph of $G^\theta_{t_{\max}-\theta}$, so it does not change the $\Delta$-slice degeneracy.

At a time strictly between two integers $t$ and $t+1$, the instantaneous graph of $\cL_\theta$ is the intersection of those at $t$ and at $t+1$, and thus has no larger maximum degree or degeneracy. Therefore all our results apply, with $k$ equal to the $\Delta$-slice degeneracy of $\mathcal{G}$, $\Delta$ equal to the largest maximum degree of a graph $G^\theta_\tau$, $m\le|\mathcal{E}|$, and $|\cT|$ at most the number of time steps.

\begin{corollary}\label{cor:delta}
A temporal graph with $n$ vertices, $|\mathcal{E}|$ time-edges, $t_{\max}-t_{\min}+1$ time steps and $\Delta$-slice degeneracy $k$ has $O(n\min(|\mathcal{E}|,t_{\max}-t_{\min}+1)k^23^{k/3})$ maximal $\Delta$-cliques, and for every fixed $\theta$ this bound is tight up to a factor polynomial in $k$. Its maximal $\Delta$-cliques can be enumerated with setup time $O(|\mathcal{E}|\log|\mathcal{E}|+|\mathcal{E}|k^2\log^3n)$ and polynomial time delay $O(k^5(k^2\log|\mathcal{E}|+\lambda(k)))$.
\end{corollary}

This should be compared with the bounds $3^{k/3}2^{k+1}n|T|$ of \cite{HMNS17} and $2^{k+1}n\min(|\mathcal{E}|,|T|)$ of \cite{BHM+19} on the larger number of time-maximal $\Delta$-cliques. In the setup time, the term $O(|\mathcal{E}|\log|\mathcal{E}|)$ accounts for the construction of $\cL_\theta$.

We now prove the tightness claimed in Corollary~\ref{cor:delta}. Take the link stream of the proof of Theorem~\ref{thm:degeneracy-count}, and let $t_{\min}=0$ and $t_{\max}=\theta+(R+1)(\theta+2)$. Every pair other than $ab$ which has a link in this stream gets a time-edge at every integer of $[t_{\min},t_{\max}]$, and the pair $ab$ gets the time-edges $\theta+j(\theta+2)$ for $j\in[R]$. Then the links of $ab$ in $\cL_\theta$ are the intervals $[j(\theta+2),\theta+j(\theta+2)]$, which are pairwise disjoint and not consecutive. Up to an increasing relabelling of the times, $\cL_\theta$ is thus the link stream of the proof of Theorem~\ref{thm:degeneracy-count}. It has at least $Rh3^p$ maximal cliques, its $\Delta$-slice degeneracy is $3p+2$, and there are $t_{\max}+1=O(R(\theta+2))$ time steps. Therefore the bound of Corollary~\ref{cor:delta} is tight up to a factor $O((\theta+2)k^2)$, and the factor $3^{k/3}$ cannot be improved.

The same reduction handles a frequency constraint, as in the first step of the algorithm of Pal \cite{Pal24}. Recall that Banerjee and Pal \cite{BP19} require every pair to interact at least $\gamma$ times in every window of length $\theta$. With the discrete convention of this section, a $(\Delta,\gamma)$-clique, for an integer $\gamma\ge 1$, is defined as a $\Delta$-clique except that, for every integer $\tau\in[a,b-\theta]$ and all distinct $u,v\in X$, there are at least $\gamma$ time-edges $(uv,t)$ with $t\in[\tau,\tau+\theta]$. Let $t_1<\dots<t_r$ be the times of the time-edges of a pair $uv$. The window $[\tau,\tau+\theta]$ contains at least $\gamma$ of them if and only if it contains $t_i,\dots,t_{i+\gamma-1}$ for some $i\in[r-\gamma+1]$, that is, if and only if $\tau\in[t_{i+\gamma-1}-\theta,t_i]$. We make every maximal set of consecutive integers of the union of these intervals, restricted to $[t_{\min},t_{\max}-\theta]$, a link of $uv$. The pair $uv$ gets at most $r$ links, which are computed in time $O(r)$ once the time-edges are sorted, and the proof of Lemma~\ref{lem:delta} applies verbatim.

\begin{corollary}\label{cor:gamma}
Let $k$ be the largest degeneracy of the graphs whose edges are the pairs having at least $\gamma$ time-edges in $[\tau,\tau+\theta]$, for integers $\tau$. A temporal graph with $n$ vertices and $|\mathcal{E}|$ time-edges has $O(n\min(|\mathcal{E}|,t_{\max}-t_{\min}+1)k^23^{k/3})$ maximal $(\Delta,\gamma)$-cliques. These cliques can be enumerated with setup time $O(|\mathcal{E}|\log|\mathcal{E}|+|\mathcal{E}|k^2\log^3n)$ and polynomial time delay $O(k^5(k^2\log|\mathcal{E}|+\lambda(k)))$.
\end{corollary}

The graphs of Corollary~\ref{cor:gamma} are subgraphs of the graphs $G^\theta_\tau$, so $k$ is at most the $\Delta$-slice degeneracy. We conclude with three remarks on the definitions of the literature. First, Banerjee and Pal \cite{BP19} also allow intervals shorter than $\theta$, for which the only window is the interval itself. Such a clique is not maximal, since its interval can be extended to an interval of length $\theta$ (recall that $t_{\max}-t_{\min}\ge\theta$). Thus the two definitions have the same maximal cliques. Second, the reduction also handles weights, as in \cite{BT25}. Assume that every time-edge has a positive weight, and that $\gamma>0$ is a lower bound on the total weight of the time-edges of a pair in every window. Then the integers $\tau$ such that the window $[\tau,\tau+\theta]$ satisfies this bound form at most $r$ maximal sets of consecutive integers, since each of them starts at $t_{\min}$ or when a time-edge enters the window. These sets are computed in time $O(r)$ by moving the window from one time-edge to the next, and the rest of the reduction is unchanged. Third, Boekhout and Takes \cite{BT25} require the interval of a clique to start and end with an interaction between two of its vertices. Their maximal cliques thus differ from ours, and our reduction does not apply to them directly. Note that they denote by $\delta=\theta+1$ the number of time steps of a window.

\section{Concluding remarks}\label{sec:conclusion}

We presented the first algorithms with polynomial time delay for maximal clique enumeration in link streams, the first one in the maximum degree and the second one in the degeneracy of the instantaneous graphs. We also proved that the maximum number of maximal cliques is $\Theta(n|\cT|3^{k/3})$ up to a factor polynomial in $k$. Several questions remain.

The second algorithm requires that the proper suffixes of the maximal cliques ending at the current time are stored in memory. We ask whether there is an algorithm with time delay polynomial in $k$ and logarithmic in $m$ which only requires space polynomial in the input. A similar question was asked for static graphs in \cite{Man19}, and it is still open.

The time delays of our algorithms contain a factor $\log m$, which comes from the searches in the data structures built during the setup phase. We ask whether the maximal cliques of a link stream can be enumerated with time delay depending only on $k$, after a setup phase polynomial in the input. Since these searches are predecessor searches, the answer may depend on the model of computation.

Our online algorithm only has a bound on its total running time. A bound on the time spent for every event would require to maintain an acyclic orientation of small out-degree with a worst-case update time. Kopelowitz et al.\ \cite{KKPS14} obtained worst-case update times, but their orientations need not be acyclic. We ask whether such an orientation can be maintained, and whether it yields an online algorithm whose time per event is polynomial in $k$ and polylogarithmic in $n$, besides the time spent to output the maximal cliques.

Fix an end time $e$ and an anchor $w$. The maximal cliques output for $(e,w)$ are cliques of the window which are maximal in $F^\tau_w$ at the threshold $\tau$ at which their last edge appears, in the increasing sequence of graphs $F^\tau_w$. This raises the following question. How many cliques of an $h$-order graph can be maximal at the moment their last edge appears, when the edges are added one by one? This number is $O(h^23^{h/3})$, and a bound $O(3^{h/3})$ would remove the factor $k^2$ from Theorem~\ref{thm:degeneracy-count}. A heuristic computer search found no insertion sequence with more than about $3.5\cdot3^{h/3}$ such cliques for $h\le 12$.

Finally, it would be interesting to extend our approach to the $\Delta$-$k$-plexes of \cite{BHM+19}, to the weighted $(\Delta,\gamma)$-cliques of \cite{BT25} whose intervals must start and end with an interaction (Corollary~\ref{cor:gamma} handles the $(\Delta,\gamma)$-cliques of \cite{BP19}, with or without weights), and to the isolated temporal cliques of \cite{MNR21}. For $k$-plexes of static graphs, a decomposition based on degeneracy orderings has been used in \cite{JM25}. The maximal cliques of Brunelli et al.\ \cite{BCGM25} are the vertex sets of our maximal cliques which are not strictly contained in a clique of some instantaneous graph. We ask whether they can be listed with polynomial time delay, or even in output polynomial time, in link streams in which a pair may interact several times. The difficulty is to test the inclusion in cliques which exist at other times. The algorithms of this paper are simple to implement with a degeneracy ordering recomputed when needed, and an experimental comparison with the algorithm of Baudin et al.\ \cite{BMT24}, whose implementation is public, would show whether the polynomial guarantee on the time per maximal clique matters on real data.

We implemented Algorithms~\ref{alg:edge} and~\ref{alg:vertex}, the latter both with degeneracy orderings and with arbitrary orders, and compared their outputs with an exhaustive search over all vertex sets and intervals on more than $10\,000$ random link streams with up to ten vertices, including streams in which many links share their end times. The outputs coincided in all cases. We also checked on these instances that every rejected candidate of Algorithm~\ref{alg:vertex} is examined after the output of the maximal clique to which it is sent by Lemma~\ref{lem:vertex-charging}, and that the constructions of Theorems~\ref{thm:degree-count} and~\ref{thm:degeneracy-count} have the stated numbers of maximal cliques. Finally, we checked the bijection of Lemma~\ref{lem:delta} against an exhaustive search of the maximal $\Delta$-cliques of random temporal graphs, and the construction that follows Corollary~\ref{cor:delta} for small values of $p$, $h$, $R$ and $\theta$. We did the same for the $(\Delta,\gamma)$-cliques of Corollary~\ref{cor:gamma}, with and without weights, and we checked that they have the same maximal cliques as the definition of \cite{BP19}.

\section*{Acknowledgements}
The author used AI assistants, in particular Claude (Anthropic), to obtain some of the results of this paper, to check them with computer experiments, and to prepare the manuscript.

\end{document}